\documentclass{IEEEtran}[10pt]

\IEEEoverridecommandlockouts

\usepackage{amsmath,                      graphicx,                      amsthm,                      amssymb,                      epsfig,                      graphicx,                      amssymb,                      url,                      bm}
\usepackage{booktabs}
\usepackage{multirow}
\usepackage{url}
\usepackage{threeparttable}
\usepackage{cases}
\usepackage{latexsym}
\usepackage{lscape}
\usepackage{arydshln}
\usepackage{balance}
\usepackage{color}

\newtheorem{prop}{Proposition}

\graphicspath{{Gated_figuresEH/}}
\title{Traffic-aware Two-stage Queueing Communication Networks:  Queue Analysis and Energy Saving}
\begin{document}

\author{
        Nan Qi, Nikolaos I. Miridakis, 
 Ming Xiao,  ~\IEEEmembership{Senior Member,~IEEE,}               Theodoros A. Tsiftsis, \IEEEmembership{Senior Member,~IEEE,} 
                \\Rugui Yao,  \IEEEmembership{Senior Member,~IEEE,}
         and Shi Jin, \IEEEmembership{Senior Member,~IEEE}   
\thanks{
Nan Qi is  with the Key Laboratory of Dynamic Cognitive System of Electromagnetic Spectrum Space, Ministry of Industry and Information Technology, Nanjing University of Aeronautics and Astronautics, Nanjing, China, 210016. Nan Qi is also with National Mobile
Communications Research Laboratory, Southeast University, Nanjing 210096,
P. R. China (e-mail: nanqi.commun@gmail.com).

 N. I. Miridakis is with the School of Electrical and Information Engineering and the Institute of Physical Internet, Jinan University, Zhuhai 519070, China. and with the Department of Electrical and Electronic Engineering, University of West Attica, Aegaleo 12244, Greece (e-mail: nikozm@uniwa.gr).

 Ming Xiao is with the School of Electrical Engineering of KTH, Royal Institute of Technology, Stockholm, Sweden (e-mail: mingx@kth.se). 
 
T. A. Tsiftsis is with the School of Intelligence Science and Engineering, Jinan University, Zhuhai 519070, China (e-mail: theo\_tsiftsis@jnu.edu.cn).
  
Rugui Yao is  with the School of Electronics and Information,  Northwestern Polytechnical University, China (e-mail: yaorg@nwpu.edu.cn)  
 
Shi Jin is with the National Mobile Communications Research
Lab, Southeast University, Nanjing, China (e-mail: 
jinshi@seu.edu.cn).

This work was supported in part by the National Natural Science Foundation of China (No.  61801218, 61871327).
Part of this work has been presented in  IEEE WCNC workshop 2019 \cite{WCNCW}. 
}
}

\maketitle

\begin{abstract}
To boost energy saving for the general delay-tolerant IoT networks, a two-stage, and single-relay queueing communication scheme is investigated. Concretely, a traffic-aware $N$-threshold and gated-service  policy are applied at the relay. As  two fundamental and significant performance  metrics, the mean waiting time and long-term expected power consumption  are explicitly  derived and related with the queueing and service parameters, such as   packet arrival rate,  service threshold and channel statistics. Besides, we take into account the  electrical circuit energy consumptions when the relay server and  access point (AP) are in different modes and  energy costs for mode transitions, whereby the power consumption model is more practical. The expected power minimization problem under  the mean waiting time constraint  is formulated. Tight closed-form bounds are adopted to obtain tractable analytical formulae  with less computational complexity. 
The optimal energy-saving service threshold that  can flexibly adjust to packet  arrival rate is determined. In addition, numerical results reveal that: 1) sacrificing  the mean waiting time not necessarily facilitates power savings; 2) a higher arrival rate leads to a greater optimal service threshold; and 3) our policy performs better than the current state-of-the-art.

\end{abstract}


\IEEEpeerreviewmaketitle

\section{Introduction} 
\subsection{Background} 
As wireless networks consecutively evolve, 
the Internet of Things (IoT), e.g., Machine-to-machine (M2M) and device to device (D2D) networks,  is emerging as a promising network formation  to fulfill the rising wireless network traffic demand worldwide \cite{what}.  
     IoT  networks  enable  devices  built-in sensors to be connected to an IoT platform and integrates packets from  different devices, therefore,    devices can communicate with each other. Low-cost and low-energy-consumption   IoT radio sensors find  wide  applications, e.g.,   smart farm/industrial monitor.   
With a rapid increase in the number of  connected devices ($25$ billion by $2020$ \cite{BR}),  
the energy consumption  in the overall IoT networks  dramatically increases and becomes a great concern.  It is demonstrated that massive  connected devices   consume  
more power than   conventional cellular and Wi-Fi networks \cite{IOT}. IoT devices are mostly small-scale and  powered by energy-constrained batteries.  {Therefore,  a low-power communication policy  is crucial for  IoT networks. }

  { Note that IoT end-users may be placed in remote or human hard-to-reach areas, where there is no reliable transmission link. As such, small-scale short-range IoT devices are unable to transmit over a 
long distance (e.g., several kilometers), which means that there are difficulties in  communicating with a distant AP.}    {  To enlarge the communication region,   relay deployment   has  been widely proposed,} especially for end-users with small transmitting power or being blocked by physical obstacles \cite{IOT}-\cite{TMRY1}. 
IoT relays  gather data from users and  relay it to destinations. 
 As one example,  M2M  relaying   
has been proposed as a  heterogeneous
architecture for  the European
Telecommunications Standards Institute (ETSI)
M2M and IEEE $802.15.4 $/$802.16$p IoT architectures \cite{IOT}.  Demos on mobile relay architectures for low-power IoT devices were presented in \cite{Demo}, \cite{giov}.   
The authors in \cite{Demo} presented a  relay-assisted video streaming application.  In  \cite{giov}, the authors investigated  a   5G IoT scenario  for   sensor data collection via drone relays. Therein, a relaying system was applied  
for environmental monitoring and agricultural applications.
 Compared with the ordinary 
IoT devices, IoT relays and access points    { (APs)}  have more resources   but  are more energy-consuming. That is, wide relay utilizations   potentially  lead  to  remarkably increasing global energy consumption \cite{ILC}.

Several prior  research works    investigated 
energy saving procedures that allow APs to dynamically  
adjust transmission functionalities according to the traffic demand \cite{ai}, and demonstrated that such an  adaptation mechanism allows a great energy saving especially in low traffic scenarios.   The  investigations   motivate  the researchers to enhance power savings at  IoT relays/APs  by applying similar methods. That is, to exploit the potential of energy saving by tracing  the traffic variations and adapting  the states of  relays/APs  accordingly.  
Inevitably, under such  a scheme,   transmitters may suffer from a  disruption in their connectivity to   servers (e.g., relays or APs). As such,   an extra  delay is incurred. Fortunately, this  is allowable for    services without real-time constraints.  
 In practice,   the delay 
tolerant network  (DTN) scheduling 
  has recently been recognized as an appropriate  technique to balance the  available  resources    and   communication tasks   without  
compromising the perceived Quality
of the Service (QoS) \cite{ILC}-\cite{CA}. There have 
been a series of research efforts on delay-tolerant IoT networks \cite{Bike}-\cite{M2M2}. 
In \cite{Bike}, the authors proposed a delay-tolerant architecture  for internet of public  bikes  by applying  reliable buffer management. 
Moreover,   an optimized delay-tolerant approach was designed for integrated RFID   IoT networks \cite{RF}, while a  
standard-based M2M platform  was configured to collect data 
from sensors with strict energy constraints \cite{M2M2}. In this paper, we focus on the intermittent transmission scheme in relay-assisted  DTN networks for the energy-saving sake without violating the delay constraint. 

\subsection{Related Works} 

The relevant researches on IoT-relaying networks can mainly be divided into two categories: $1)$ network architecture and routing protocol designs \cite{Demo}; $2)$   resource management, including relay number  determination \cite{AA}, \cite{DD},  placement \cite{CA}, device-relay-channel association \cite{AA}.  In \cite{AA}, \cite{DD},   optimal relay  placements and device-relay-channel associations were determined to  achieve energy efficient  uplink transmissions. The relay placement problem under  a simplified  delay constraint was investigated in \cite{CA}. 
Nevertheless, the above relay research works did not consider packet transmitting scheduling in time domain. 

 
Quite recently,  traffic-aware   scheduling problems in conventional cellular networks   have grasped considerable attentions  \cite{JKHL}-\cite{LS}. By configuring the lightly-loaded relays  into their sleeping  state,  the  working duration of   relays/APs is reduced, bringing   the benefit of a lowered basic electrical  consumption.  
In \cite{JKHL}, an  interesting aerial flying network infrastructure was envisioned, which enables the network topology to be  dynamically reconfigured   to match the users' traffic demands. Similar work can be found in   \cite{ILC}, \cite{LS} and \cite{JWSZ}, where   nodes sleeping and user transmission scheduling  were investigated, aiming at diminishing the energy consumptions by exploiting both temporal and spatial variation of traffic loads.   
In addition, some important and pertinent research works  on   traffic-aware communication networks were performed in \cite{WJ2016}, \cite{WJ2013}, where an 
 $N$-threshold base stations service rule has been applied to one-hop transmissions. The authors illustrated that one optimum threshold  exists in terms of the energy-delay trade-off.  However, the  design is infeasible in the relay-assisted transmission scheme.
 
Inherently, traffic-aware relay-assisted two-stage transmissions implemented by adaptively configuring the relays/APs' service states  can be  modelled as a two-stage queueing system. 
To the best of our knowledge, very few researches are performed for  such network settings. 
 The most relevant works are \cite{DB}-\cite{KT}, where   two-stage M/G/$1$ queueing systems have been studied.   An $N$-threshold relay service rule was also applied.  
Although   \cite{DB}-\cite{KT} presented significant  insights for two-stage queueing networks,   they  investigated   the scenarios where \textit{the first stage was in the batch-service mode and cannot apply to wireless relaying system.  }
In relaying wireless communication networks, however,  users  are mainly  served  in an individual mode during  each hop.   Some   research works on  two-stage queueing networks with individual-service mode are \cite{CA}  and \cite{JL}. However, they supposed a  relay always-on service policy, and failed to take into account the power-saving relevant queueing service  design.

The traffic-aware $N$-threshold transmission scheme  inherently differentiates from the previously published literature  \cite{WJ2016}-\cite{JL}. In particular, the mean waiting time, as a key performance  metric,  is typically composed of the queuing delay,  transmission delay as well as processing 
delay. Rigorous mathematical analysis for the mean waiting time along  with the overall  power consumption  that includes the electrical consumption has not been investigated yet. 
The mean waiting time, in particular, cannot be simply obtained from the widely adopted Little's Law, since it is closely related with the service stage when the packet arrives.   Furthermore, the expected power consumption is related to the states of the relay and AP along with the time span that they remain in those states. Hence, the power  consumption   model is also different from those in the published literature \cite{WJ2016}-\cite{JL}.

\subsection{Contributions} 
 We focus on seeking low-power solutions    while guaranteeing  the transmitters' delay performance for delay-tolerant networks.  One fundamental and critical issue is to investigate how   two key  metrics, i.e., energy consumption and mean waiting time are related with   packet arrival rate, queueing service rate,  the service threshold, \emph{etc}.    
    { To the  best of our knowledge,   relay-aided two-stage queueing communication  service   is still an open issue and  different from the state-of-the-art.}  Three critical and technically challenging issues are:  1)   derivation  of the expected queuing length in different service  stages;   2) formulation  of the mean waiting time; and 3) determination of energy-saving service threshold that adapts to various mean waiting time tolerance and packet arrival rates. 

Specifically, our main contributions are listed as follows:  

\begin{description}

\item[$\bullet$] A more general and practical two-hop queueing system is  investigated. Two random processes are considered, including the ergodic fading channels and Poisson arrival process. Additionally, the power consumption model is more practical in the sense that we take into account the electrical circuit energy consumptions when nodes are in different modes as well as power costs for mode transitions. 
 
\item[$\bullet$] As a key metric, the  mean waiting time is first  derived.  The mean waiting time hence can be measured mathematically.  Based on that, the long-term expected power consumption is formulated. Then, how they are affected by the system parameters, namely   packet arrival rate,   channel statistics, and relay service threshold     { is}  mathematically  analysed.

\item[$\bullet$]  By adopting tight performance bounds  based on Jensen's inequality,  complicated  manifold integral computations are  avoided.    {  This    
 method  enables} our results to be   derived   in tractable closed-formulae, and  evaluated quite easily. Hence, the overall computational complexity is greatly reduced.

\item[$\bullet$] The optimum service threshold that can be flexibly adjusted to   packet arrival rate is  determined such that   the long-term expected power consumption  is minimized without violating the delay constraint.  

\end{description}

The rest of the paper is organized as follows. In Section II, we present the system model. Queueing behaviour analysis is given in Section III. In Section IV, the considered power minimization problem is formulated and analyzed. Problem approximation and determination are  provided in Section V. The analytical and simulation results are presented in Section VI. Finally, Section VII concludes this paper.

   \begin{table}[htbp]
\caption{Notations}
\centering 
\begin{tabular}{c p{150pt}}
\hline
\textbf{Notation}                                                                                                             & \textbf{Description}            \\ \hline\hline 
$N_0$                                                                                                                       & the number of accumulated packets at the beginning of the first service stage (FSS)      \\ \hline
$ N_1$ & the number of arrivals during $N_0$ packets' FSS  \\ \hline
$ N_2$ &   the number of accumulated  packets  at the end of  the second service stage   \\ \hline
$N_0(z)$, $N_1(z)$, $N_2(z)$                                                                                                                 & PGFs of $N_0$,  $N_1$, $N_2$           \\ \hline
$\mathcal{L}_{{T}_1}(\theta)$, $\mathcal{L}_{{T}_2}(\theta)$ & Laplace-Stieltjes transforms (LSTs) of  ${T}_1(t)$, ${T}_2(t)$      \\ \hline
\end{tabular}\label{Parameter_}
\end{table}

\section{System Model}  
As depicted in Fig. \ref{mobile_relaying_systemmodel},  we consider a two-stage and single-relay queueing IoT  communication system consisting of    users, a relay station and an AP that is connected to the Internet. The users can be  distributed   sensors, actuators, cameras for smart farm/industrial/forestry applications. The AP can be a   data integration point in the practical IoT scenarios.   
     {Besides, we assume a central-controlled network framework. The central node can be the relay  or one dedicated node that  is capable to communicate with each transceiver in the network. Once a packet is generated at one specific user, the user notifies the central node of its existence via a short message (e.g., beacon message).  
The   packet arrivals  at the central node follow the Poisson distribution with average rate $\lambda$    {(measured in packets/ms)}.  Note that the value of  $\lambda$ varies with the statistic  traffic load. }

     {Uplink data transmissions are considered, i.e. the users intend to upload their packets to the AP. However, }  the direct wireless communication links between the users and AP are unavailable due to the long-distance, severe fading or physical obstacles. Additionally, no  Internet connectivity is 
available. As such, a relay station is applied to wirelessly collect data from sensors and forward  them to the AP.    Specifically, a two-stage service is conducted to deliver the  arriving packets to the AP via the  the relay. The relay is equipped with a single antenna and acts as a server in the queueing system. The AP-Internet transmissions   can be realized with high-capacity wired  cables. We only elaborate on the two-stage wireless transmissions, while AP-Internet transmissions are ignored in the paper.

     {
The packet queue is managed by the central node.   
   In what follows, a sub-cycle is first introduced,  which is composed  of   two service stages, i.e.,   the first service stage (FSS) and  second service stage (SSS), respectively.  In the FSS, the users transmit their packets to the relays sequentially.  In the SSS, the packets received at the relay are successively forwarded to the AP.   }
 We take an arbitrary sub-cycle below as an example to  illustrate the service model.  Further, the definition of a regeneration cycle is introduced, which consists of several consecutive sub-cycles.    {Note that    the   notations are  listed  in Table \ref{Parameter_}.}
  
\begin{figure}[!htpb]
\centering
\includegraphics[width=0.5\textwidth]{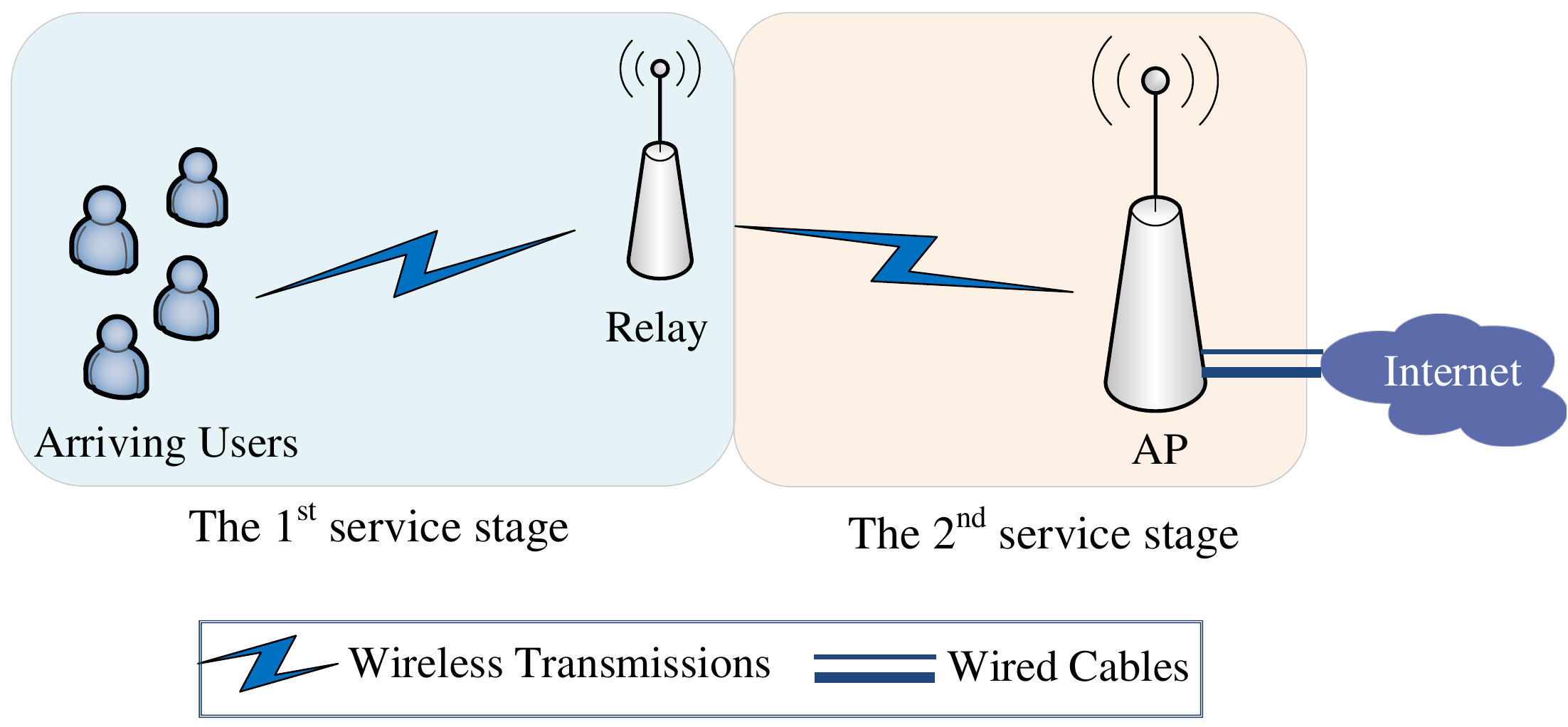}
\caption{Two-stage queueing service model.}\label{mobile_relaying_systemmodel}
\end{figure}

\begin{figure}[!htpb]
\centering
\includegraphics[width=0.5\textwidth]{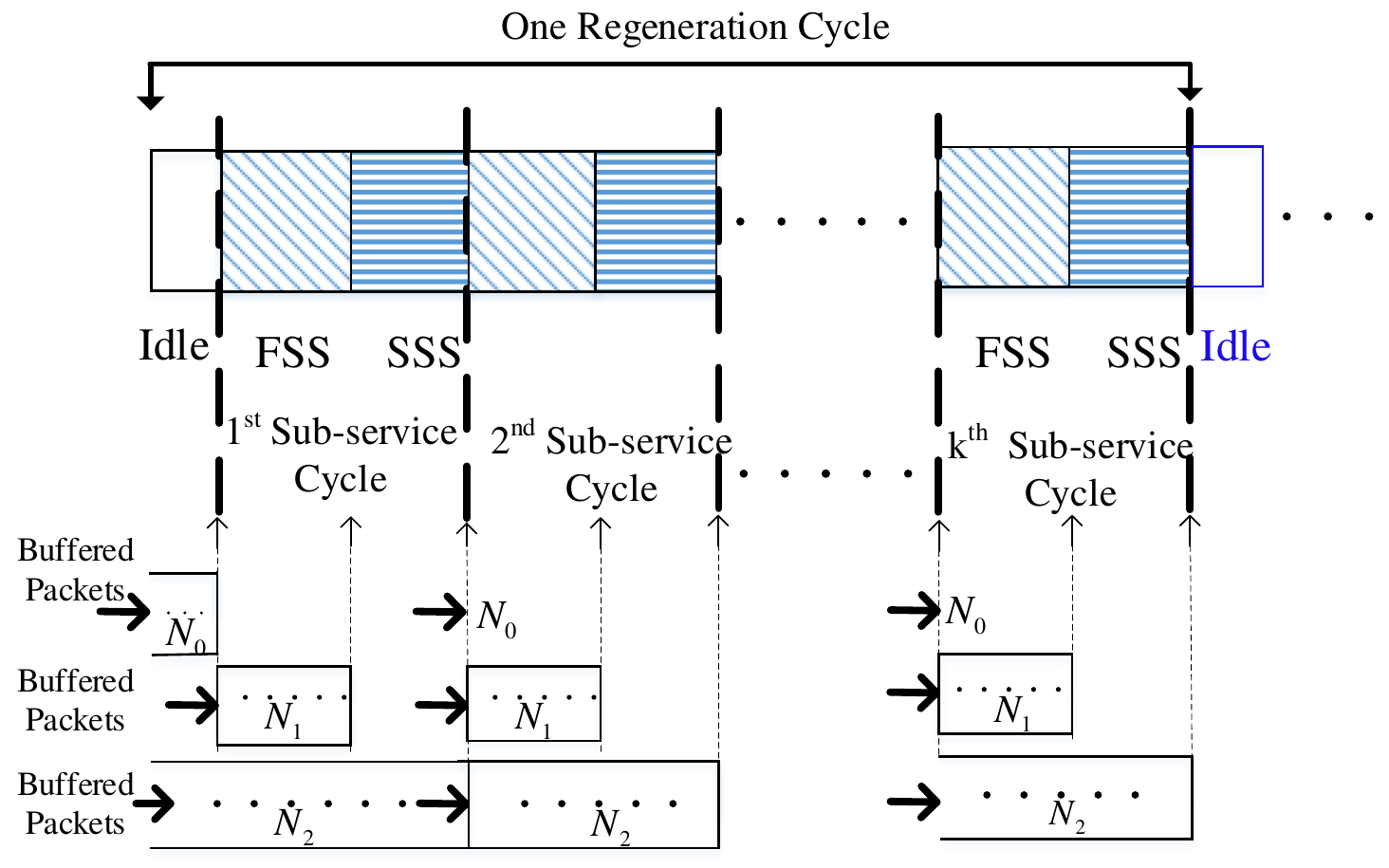}
\caption{One Regeneration Cycle Model Illustration.}\label{Regeneration_Cycle}
\end{figure}  
  
\subsection{The First Service Stage (FSS)}  \label{FSS}

The $N$-threshold policy is applied.      {   The detailed policy is illustrated as follows.}       {Let $ N_2$  represent  the number of accumulated  packets  at the end of  the second service stage (which will be illustrated in Subsection \ref{SSS}).}   
         {Only when  $N_2>0$, the relay server switches immediately into the receiving mode at the end of the SSS\footnote{ When the number of accumulated packets satisfies a pre-set rule, the relay    sends the packet-transmitting-request to the  waiting users. After that, the relay immediately switches  into the receiving mode. }.  In case that $N_2=0$, the relay server remains in the listening  mode  till  the number of newly arrived  packets reaches the predefined threshold $N$, i.e.,} 
         { 
\begin{equation}
N_0  = \left\{ \begin{array}{l}
 N_2 ,\quad \text{if}\;N_2  > 0, \\ 
N,\quad \;\;\text{otherwise}. \label{NEQ_0}\\ 
 \end{array} \right.
\end{equation}}
   { Based on our definition of $ N_0 $, we can observe that $ N_0 $
is inherently a random variable, since $N_2$  brings randomness to it.}  
Besides, note that during the FSS,  the AP remains listening.
   { Specifically, 
in the listening mode, a transceiver is active and ready but not currently  receiving or transmitting packets.   Listening to the channel and  executing clear channels assessment are the main activity in the listening mode}\footnote{The AP and relay is not completely shut off   because  it may be quite challenging to entirely
turn off the AP and then require it  switch to the receiving mode immediately for the next second service stage. Nevertheless,    listening   
mode control within a short time span is feasible, as investigated in   \cite{PF}.}.

 Additionally, a  gated-policy is adopted at the relay. Specifically, only the   packets that are already present  (before the FSS starts) will be served in the current sub-cycle, while those arrive during the current FSS have to wait for being served in the next sub-cycle.

 In the MAC layer, a schedule-based MAC protocol, i.e., time division multiple 
access (TDMA) is applied to coordinate  packet transmitting among the users. Specifically, the packets are individually served by the relay server in the first come first served (FCFS) rule. Note that packet collision is avoided under the TDMA-based MAC protocol.  In the FSS, the active relay keeps receiving   packets, decoding  and storing them in its data buffer. It is assumed that the buffer size is large enough to store the packets, which is reasonable since in practice, only a limited number of packets are accumulated in a sub-cycle; otherwise, the waiting time will be intolerable.

In the FSS, the service rate is
 \begin{equation}
C_1 (h_{1} ) = B\log \left(1 + \frac{{|h_{1} |^2 p_1 }}{{\sigma _{1}^2 }}\right), \label{C1} 
\end{equation}
where  $B$ is the bandwidth;        {$p_1$ is the data transmitting power of the packets};  $\sigma_1^2$ is the noise power; $h_{1}\sim \mathcal{CN}(0, \sigma _{{h_{1}}}^2)$ is the  channel gain;
 $ \mathcal{CN}(m,v)$ denotes a complex Gaussian random variable with $m$ and $v$ standing for the mean value and variance, respectively; 
 $|h_{1} |$ is the amplitude of $h_{1}$ and follows the Rayleigh distribution;  $\sigma _{{h_{1}}}^2$ is the variance of   Rayleigh distribution.  Rayleigh  fading channels are assumed to be independent and identically distributed (\emph{i.i.d.}) for mathematical tractability and drawing  insights.

      {The $i$th arriving packet is indexed by packet  $i$}.       { The service time span for individual packet (or individual service time span  for short), denoted by $\{{T_{1,i} },\; i=1,\,2\;\cdots,  \}$}, in the FSS  are assumed to be a stationary \emph{i.i.d.}  process with a distribution function, ${T_1}(t)$. Note that such an assumption can be realizable in practice. 
One feasible way is to flexibly match  each packet's  transmitting power  with the channel stochastic parameters, such that the service rate   is a stationary \emph{i.i.d.} process. 
 It is further assumed that  the length of each individual packet is fixed and denoted    {as $l$}\footnote{  If the length is not fixed, we can split longer packets into shorter ones such that the length of each packet is the same.}.    The individual service time is obtained as
\begin{equation}
{T_{1} } =l/{C_1 (h_{1} )},\label{T11}
\end{equation}
while its distribution function, denoted as ${T_{1} }(t)$, is  
\begin{equation}
{T_{1} }(t) ={\Pr}\Big\{\frac{l}{C_1 (h_{1} )}<t\Big\}.
\label{T_i}
\end{equation}

Further, by combining with \eqref{C1},  ${T_{1} }(t)$ can be presented as follows:
\begin{equation}
{T_{1} }(t) =\exp \Big( - \frac{{\big(\exp (\frac{l}{{Bt}}) - 1\big)\sigma _1^2 }}{{\sigma _{h_1 }^2 p_1 }}\Big).
\label{T2_i}
\end{equation}

\subsection{The Second Service Stage (SSS)} \label{SSS} 
Once  the FSS is completed, the   SSS starts and the AP switches from   the listening mode to the receiving mode.
 All the packets stored in the relay's buffer will be forwarded to the AP according to the FCFS rule.  The individual service time span in the SSS, denoted by $\{{T_{2,i} },\; i=1,\,2\;\cdots,  \}$, is assumed to be a stationary \emph{i.i.d.} process with a  distribution function ${T_2}(t)$.  
 Note that   new arrivals during the SSS will wait  for the next ``FSS+SSS"  sub-cycle   service.
 
The service rate in the SSS, denoted as $C_2 (h_{2})$, can be given by
 \begin{equation}
C_2 (h_{2} ) = B\log \left(1 + \frac{{|h_{2} |^2 p_2 }}{{\sigma _{2}^2 }}\right),\label{C2} 
\end{equation}
where $p_2$ is the data transmitting power at the relay;  $h_{2}$ and  $\sigma_2^2$ are defined respectively similar to  $h_{1}$ and $\sigma_1^2$, and correspond  to the parameters for the relay-AP channel  in the SSS\footnote{${T_{2} }$ and ${T_{2} }(t)$ can be directly obtained  by substituting index ``$1$" in \eqref{T11} and  \eqref{T_i} with ``$2$".  The details are omitted here.}    {.}

Correspondingly, the individual service time in the SSS can be  obtained as
\begin{equation}
{T_{2} } =l/{C_2 (h_{2} )}.\label{T22}
\end{equation}
         {$\{{T_{2,i} },\; i=1,\,2\;\cdots,  \}$ is assumed to be a stationary \emph{i.i.d.} process with a  distribution function ${T_2}(t)$. The detailed expression for ${T_2}(t)$ can be obtained by changing ``1"  in  \eqref{T2_i} 
  into ``2". We omit its expression here. } 

\subsection{The Regeneration Cycle Model}

It is noteworthy that intermittent wireless transmission scheme is  applied. Specifically, at  the  end of the SSS, depending on the value of $N_2$, 
 the relay server decides whether it starts the  next   sub-service  cycle  immediately or not.          {As described in Section \ref{FSS}, if  $N_2\geq 0$, the next  sub-service cycle will start, resulting in   consecutive sub-service cycles; otherwise,  the sub-service  cycle will be followed  by an idle period.  }

As depicted in Fig.~\ref{Regeneration_Cycle}, a regeneration cycle is introduced and defined as the serving process between two idle states, consisting of an idle period followed by several consecutive sub-service  cycles. The number of sub-service cycle  is a random variable, denoted as $k$. Statistically, $k$ also represents the number of   sub-cycle services before the service turns into   idle mode,  indicating that $k$ follows the geometric distribution. Hence, we have  
\begin{equation}
 \mathbb{E}\{k\}=1/{\pi_0}, \label{kpi0}
\end{equation}
where $\mathbb{E}\{\cdot\}$ is the expectation operator;
 $\pi_0={\Pr}\big\{N_2=0\big\}$, and equals to  the probability that  the relay switches from the transmitting mode to  the listening mode.

For the above queueing service model,   we will analyse its key measures in the sequel, including the regeneration cycle time span, queue lengths at different service  stages, based on which the mean waiting time is  gradually derived. 

\section{Regeneration Cycle Measures Analysis}

\subsection{The Regeneration Cycle Time Span}

   {
Denote the number of  packet arrivals  in one complete regeneration  cycle as $\Gamma$.     
$\mathbb{E}\{ \Gamma \}$ is the expected number of arrived packets in one complete regeneration  cycle. In what follows, we respectively calculate the expected number of packet arrivals during the   idle, FSSs and SSSs.    }  

   {First,   note that in the idle mode, the relaying service will not start until there are $N$ packets waiting in the queue.  Additionally,
 there are in total  $\Gamma$ arrived packets, each taking an average of   $\mathbb{E}\{ T_1 \}$  and  $\mathbb{E}\{ T_2 \}$ seconds for the FSS and SSS,  respectively.  According to Wald's equation \cite{Book_queue}, it takes an average of  $\mathbb{E}\{ \Gamma \} \mathbb{E}\{ T_1 \}$  and $\mathbb{E}\{ \Gamma \} \mathbb{E}\{ T_2 \}$ seconds  for serving an average of  $\mathbb{E}\{ \Gamma \}$ packets in the FSSs and SSSs  of one complete    regeneration  cycle,  respectively. Based on the delay cycle property in Section $1.2$ in  \cite{Book_queue}, we achieve  that   the total  expected numbers of packet arrivals in the  idle stage, FSSs, and SSSs  are $N$, $\lambda  \cdot\mathbb{E}\{ \Gamma \} \mathbb{E}\{ T_1 \}$ and $\lambda  \cdot\mathbb{E}\{ \Gamma \} \mathbb{E}\{ T_2 \}$, respectively (for detailed derivations, please refer to Appendix \ref{Wald}).    $\mathbb{E}\{ \Gamma \}$  can be given by}
\begin{equation}
\mathbb{E}\{ \Gamma \}  = N + \lambda  \cdot \mathbb{E}\{ \Gamma \} \mathbb{E}\{ T_1 \}  + \lambda  \cdot \mathbb{E}\{ \Gamma \} \mathbb{E}\{ T_2 \} \label{E_Gamma}.
\end{equation} 

   {Let $\mathbb{E}\{ T_0 \}$ represent the expected idle duration. Since the packet arrival follows the Poisson distribution, $\mathbb{E}\{ T_0 \}$ is related with $N$  as follows:
\begin{equation}
\mathbb{E}\{ T_0 \}=\frac{N}{\lambda}.\label{ET0}
\end{equation}
}
 
From \eqref{E_Gamma},  we obtain that the closed-form of $\mathbb{E}\{ \Gamma \}$  can be expressed as 
\begin{equation}
\mathbb{E}\{ \Gamma \}  = \frac{N}{{1 - \lambda \mathbb{E}\{ T_1 \}  - \lambda \mathbb{E}\{ T_2 \} }}. \label{E_gamma}
\end{equation}

   {Let us denote  $X_i=|h_{i}|^2$, $Y_i=\frac{l}{C_i (X_i)}$ ($i=1,\;2$), for convenience.  Since $X_i=|h_{i}|^2$ follows   
an exponential \emph{p.d.f.} \cite{Book_queue}, such that
\begin{equation}
g_{X_i}(x) = \frac{1}{{\sigma _{h_i}^2}}\exp ( - \frac{x}{{\sigma _{h_i}^2}}).\nonumber
\end{equation}
Further, $\mathbb{E}\{ T_i \}$  ($i\in\{1,2\}$) in \eqref{E_Gamma} and \eqref{E_gamma} is  calculated as} 
\begin{equation}
\mathbb{E}\{ T_i \}=\int_0^{ + \infty } {\frac{l}{{C_i (x)}}} g_{X_i } (x)dx\footnote{Note that in practice, $X_i=|h_i|^{2}>0$. For mathematical computation, we bound $x$ as $x\geq\delta$, where $\delta$ takes a considerably small value, e.g., $10^{-6}$.}. \label{eTI}
\end{equation}

Let $T_{RC}$ represent the  time span of one  regeneration cycle.   $\mathbb{E}\{ T_{RC}\} $ can be derived as    {(for detailed derivations, please refer to Appendix \ref{Wald})}
\begin{equation}
\mathbb{E}\{ T_{RC}\}  = \frac{\mathbb{E}\{ \Gamma \}}{\lambda}. \label{E_g2amma}
\end{equation}

    { Direct derivations of the mean waiting time are considerably challenging. Note that the mean waiting time can be obtained by taking the derivatives of the LST of the distribution of the waiting time.   The LST of the distribution of the waiting time can be achieved from  the probability generating functions (PGFs) of the  queue length. }

    {In what follows, we will respectively formulate the PGFs  of the FSS and SSS queue lengths, based on which the mean waiting time is further deduced.} 

\subsection{PGFs of the FSS and SSS Queue Lengths}

     {Let $ N_1$ represent the number of arrivals during $N_0$ packets' FSS, and 
$N_0$                                                                                                                         represent  the number of accumulated packets at the beginning of the first service stage (FSS). 
 Let  
$N_0(z)$, $N_1(z)$, $N_2(z)$                                                                                                                 represent the  PGFs of $N_0$,  $N_1$, $N_2$.  In the following, the deductions of  $N_0(z)$, $N_1(z)$ and $N_2(z)$  are illustrated first. On this basis,   $\pi_0$  is   derived.}
 
     {According to the relations between the LST and PGF of the Poisson arrival process  (interested readers are referred to   Page $5$ of \cite{Book_queue}),  
  $N_0(z)$, $N_1(z)$ and $N_2(z)$   can be   achieved from 
the  LST of the distribution function  ${{T}_i}(t)$, $i=1,\; 2$. In this subsection, we first derive the  LSTs. }
 
Let us denote  $Y_i=\frac{l}{C_i (X_i)}$ ($i=1,\;2$), for convenience. Let $g_{Y_i}(y)$ represent the probability density function (\emph{p.d.f.}) of $Y_i$.      {By definition of LST (for interested readers, please refer to Eq. (1.4a) in \cite{Book_queue}), we obtain the  LST of the distribution function  ${{T}_i}(t)$ as}  
\begin{equation}
\mathcal{L}_{{T}_i}(\theta)=\int_0^{ + \infty } {\exp ( - \theta y)} g_{Y_i} (y)dy. \label{eq2}
\end{equation}
$Y_i=\frac{l}{C_i (X_i)}$ monotonically decreases with $X_i$, then $g_{Y_i} (y)$ is  obtained as  $
g_{Y_i} (y) = \frac{g_{X_i}(x)}{|\nabla _x y|}$ \cite{pdfconversion},
where ${\nabla _x y}$ is the derivative of $y$ \emph{w.r.t.} $x$.
Finally, \eqref{eq2} can be rewritten as 
\begin{equation}
\mathcal{L}_{{T}_i}(\theta)=\int_0^{ + \infty } {\exp \Big( - \theta \frac{l}{C_i (x)}\Big)} {g_{X_i}(x)}dx. \label{e2q2}
\end{equation}

Note that the PGFs of the number of packet  arrivals during one packet's   first and second service stage   can be respectively  calculated as  ${\cal L}_{T_1 }  (\lambda  - \lambda z)$ and ${\cal L}_{T_2 }  (\lambda  - \lambda z )$ (see page $5$ of \cite{Book_queue})). Then, during the overall  individual service time of $N_0$                                                                                                                        accumulated packets,  the PGFs of the number of packet  arrivals can be given as 
\begin{equation}
N_1 (z) = \sum\limits_{i = 0}^{ + \infty } {\Pr \{ N_0  = i\} } \Big\{{\cal L}_{T_1 } (\lambda  - \lambda z)\Big\}^i.\label{N10Z}
\end{equation}
   { According to the definition of PDF, we rewrite $N_1(z)$ as }
\begin{equation}
N_1 (z) = N_0 \Big({\cal L}_{T_1 } (\lambda  - \lambda z)\Big).\label{N1Z}
\end{equation} 
Note that \eqref{N1Z} also reveals the relationship between $N_1(z)$ and $N_0(z)$.

Further, by  definitions of $N_2$ and $N_0$,  the  relations between  $N_2(z)$ and $N_0(z)$  are 
         {\begin{align}
N_2(z) &= N_0 \Big({\cal L}_{T_1 } (\lambda  - \lambda z)\Big)\cdot N_0 \Big({\cal L}_{T_2 } (\lambda  - \lambda z)\Big), \label{N2Z}\\
N_0(z) &= N_2(z)-\pi_0+\pi_0 z^N.\label{N0Z}
\end{align} }

It can be noticed from \eqref{N1Z}-\eqref{N0Z} that $N_0(z)$ is the kernel of the derivations on $N_1(z)$ and $N_2(z)$. For $N_0(z)$, we   have the following Proposition: 
\begin{prop}  \label{Theorem1}
For the two-stage service, where   the individual service time span is  {i.i.d.},  $N_0 (z)$ can be given as 
\begin{align}
N_0 (z) = &\bigg(\bigg [{\Big( {{{1 }}-g_{n  - 1}(z)} \Big)}^2  - g_{n  - 2}(z)\bigg] ^2  -  \cdots  - g_1(z)\bigg) ^2\nonumber\\
&- g_0(z),  \label{n2on}
\end{align}
where $n\in \mathbb{N}^+$ and $n \to  \infty$;   $g_{n}(z)$ ($\forall n$) is defined in  \eqref{g_n}.

Additionally, $N_0 (z)$ absolutely  converges for any $|z|<1$. $N_1 (z)$ and $N_2 (z)$ can be further obtained   from \eqref{N1Z} and \eqref{N2Z}. Their detailed  expressions are omitted here.
\end{prop}

\begin{proof}
The proof is presented in Appendix  \ref{appendixconvexproof}.
\end{proof}
 
     { \textbf{Proposition \ref{Theorem1}} provides a  general PGF formula. 
The derivations in \eqref{g_n}-\eqref{iteration} indicate  that $N_0 (z)$ is a function \emph{w.r.t.}   $N$, $\pi _0$,  and the queueing service parameters in \eqref{e2q2}, including $l$, $B$, and ${\sigma _{h_i}^2}$  $(i=1,2)$. }

   { In particular,} \textbf{Proposition \ref{Theorem1}} forms a fundamental  basis for all subsequent derivations and provides   two following  clues.  
 On  the one hand, the PGF property that $N_0 (0)= \Pr \{ N_0  = 0\}$ provides  a clue to investigate the relationship among different parameters,  e.g., $N$, $\pi _0$ and the queueing service parameters in  \eqref{e2q2}. On  the other  hand,  $N_0 (z)$ can be directly obtained once $z$ is given, which is useful in deriving the LST of different service stages.  In particular, by replacing ${\cal L}_{T_i }  (\theta)$ with $z$, we can easily obtain the LST of the  distribution function  of  $N_0$ packets' service time span (either in the FSS or the SSS, or a smaller-scale    stage). That is,  $N_0 (z)$  also provides  clues in deriving the mean service duration.

 As can be found in \textbf{Proposition \ref{Theorem1}}, $\pi _0$ is a fundamental and critical parameter in obtaining  $N_0(z)$ in \eqref{n2on} as well as subsequent derivation of the mean waiting time.  
On the other hand, according to \eqref{kpi0}, $\pi_0$ determines the mode transmission frequency at the AP. It  is closely related to  power consumption  at the AP. For $\pi_0$, \textbf{Proposition \ref{p00ai}} is provided  as follows: 

\begin{prop}  \label{p00ai}
For the two-stage queueing service, where   the individual service time span is  {i.i.d.}, $\pi_0$  is 
 \begin{equation}
 \pi _0=\bigg\{（N_0 \Big( {{\cal L}_{T_2 } (\lambda)} \Big)\bigg\}^2. \label{Pi0}
 \end{equation}
\end{prop} 
\begin{proof}
The proof is relegated in Appendix  \ref{appendixconvexproof2}.
\end{proof}

 \textbf{Proposition \ref{p00ai}} sheds light on  the relationship between $\pi_0$ and various system parameters, including the   packet arrival rate,   channel statistics,  and relay service threshold.   Although \eqref{Pi0} is not in an explicit-form, $\pi_0$ can be numerically determined efficiently with the approximation method that will be illustrated in Section V. 

\subsection{Mean Waiting Time}
 Throughout the paper, the waiting time represents the duration from one packet joins the queue system until its turn to  be served in the SSS, which is   composed of the  queuing delay, processing delay as well as transmission delay. Their specific meanings are presented as follows.

The queuing delay represents the time span from the instant  a packet is generated (equivalent to the instant when the packet enters the queue) till the instant    { it starts to be transmitted}. Clearly, it is dependent on  the traffic intensity as well as the link service rate.  
The processing delay is the time duration spent for (de)coding, (de)modulation, etc, which is ignorable.    On the other hand, under the assumption  that the CSI is known at the transmitters and with a maximum rate of the Shannon capacity, the packet can be correctly received and  retransmissions are avoided, theoretically. Each packet's  transmission delay can be obtained from \eqref{T11} and \eqref{T22}.

In this section,  the  queuing delay as well as transmission delay are derived. One  observed packet  has to wait for its data being forwarded to the AP in the following three cases: 
\subsubsection*{$\bullet$ Case $1$:  the observed packet arrives when the server is idle}
 
 In case $1$,  the observed packet,  denoted by $u_{i'}$ ($i \in \{1,2,\cdots, N\}$)  will wait 
 
\begin{description}
\item[(i)] the remaining idle period, plus

\item[(ii)] the time span of the  FSS, plus

\item[(iii)] the  overall SSS time span of packets that arrive prior to  $u_{i'}$.
\end{description}
\begin{align}
W_q(\theta {\rm{|}}\,case\,1) =& \sum\limits_{i = 1}^N \Pr \{ i' = i\}  \cdot \big\{{{\cal L}_{T_{INT} }(\theta )} \big\}
^{N - i}\nonumber \\ 
&\cdot \big\{{{\cal L}_{T_1 } } (\theta )\big\} ^{N} \big\{ {{\cal L}_{T_2 } }(\theta )\big\} ^{i - 1}, \label{waitcase1}
\end{align}
where ${\Pr \{ i' = i\}} $ represents the probability that the observed packet ranks the $i$th among $N$ packets  with respect to (\emph{w.r.t.}) the arrival time, and can be calculated as $
\Pr \{ i' = i\}  = \frac{1}{N};$ 
${{\cal L}_{T_{INT} }(\theta )}$ is LST of the distribution function of the time interval between any two neighbouring arrivals. For  Poisson arrivals, ${{\cal L}_{T_{INT} }(\theta )}$ can be computed as follows: 
\begin{equation}
{{\cal L}_{T_{INT} }(\theta )}=\int_0^{ + \infty } \exp( - \theta x)\Big(  \lambda \exp( - \lambda x)\Big)dx
=\frac{{\lambda }}{{\theta  + \lambda }}.
\end{equation}

\subsubsection*{$\bullet$ Case $2$: The observed packet arrives during the FSS}
 
Let $T'$  represent the total individual service time span of $N_0$ packets  in the FSS. The LST of the distribution function of  $T'$  can be expressed as 
\begin{equation}
{\cal L}_{T'} (\theta ) = N_0 \Big( {{\cal L}_{T_1 } (\theta )} \Big).\label{Ttheta}
\end{equation}

One observed packet arrives during the FSS will wait

\begin{description}
\item[(i)] the remaining time of $T'$ period, plus

\item[(ii)] the overall SSS time span of $N_0$ packets in the currently undergoing sub-service cycle, plus

\item[(iii)] the overall FSS time span of $N_0$ packets in  the next  sub-service  cycle, plus

\item[(iv)] the overall SSS time span (in  the next  sub-service  cycle) of the packets that arrive during the elapsed time of $T'$ period.
\end{description}

Note that (i)+(iv) can be obtained according to the joint LST property of the elapsed and remaining service time distribution function, i.e., 
\begin{equation}
X_ \pm ^{(case2)}(\theta) =\frac{{{\cal L}_{T'} \left( {\lambda  - \lambda {\cal L}_{T_2 } (\theta )} \right) - {\cal L}_{T'} (\theta )}}{\left(\theta  - \lambda  + \lambda {\cal L}_{T_2 } (\theta )\right)\mathbb{E}\{T'\}}.\label{jointcase2}
\end{equation}
For detailed derivations, interested readers can refer to Eq.(1.54b) in \cite{Book_queue}. Further, according to the  property of the PGF, namely,  PGF of sums of independent random variables equals to the  multiplication of PGF of each independent random variable (which is referred  to as the Sum-Multiplication property in the following), the LST of  the distribution function of waiting time in  Case $2$ is finally given as
\begin{align}
W_q(\theta {\rm{|}}\,case\,2) =&X_ \pm ^{(case2)}(\theta)\cdot N_0\Big({\cal L}_{T_2} \left( { \theta } \right)\Big)\cdot N_0\Big({\cal L}_{T_1} \left( { \theta } \right)\Big), \label{wait2}
\end{align}  
where   $\mathbb{E}\{T'\}$ can be obtained by differentiating \eqref{Ttheta} \emph{w.r.t.} $\theta$. Specifically, we have
\begin{equation}
\mathbb{E}\{T'\}=-\nabla_\theta {{\cal L}_{T'} (\theta )}\Big |_{\theta=0}. \label{ET}
\end{equation} 

\subsubsection*{$\bullet$ Case $3$: the observed packet joins the waiting queue during the SSS}

Let $T''$  denote the total SSS time span of $N_0$ packets. The LST of the distribution function of $T''$ is represented as
\begin{equation}
{\cal L}_{T''} (\theta ) = N_0 \Big( {{\cal L}_{T_2 } (\theta )} \Big).\label{TPIPP}
\end{equation}

 One observed packet arrives during the SSS will wait
\begin{description}
\item[(i)] the remaining time of $T''$ period, plus

\item[(ii)] the overall FSS time span of $N_0$ packets in  the next  sub-service  cycle, plus

\item[(iii)]  the overall  SSS time span of packets that arrive during the  FSS   of the currently undergoing   sub-service  cycle, plus

\item[(iv)] the overall  SSS time  span  of the packets that arrive during the elapsed time of $T''$ period.
\end{description}

Similar to \eqref{jointcase2},  (i)+(iv) can be obtained according to the joint LST property of the elapsed and remaining service time distribution function, which can be expressed as 
\begin{equation}
X_ \pm ^{(case3)}(\theta) =\frac{{{\cal L}_{T''} \Big( {\lambda  - \lambda {\cal L}_{T_2 } (\theta )} \Big) - {\cal L}_{T''} (\theta )}}{\Big(\theta  - \lambda  + \lambda {\cal L}_{T_2 } (\theta )\Big)\mathbb{E}\{T''\}} .\label{jointcase3}
\end{equation}

The LST of the distribution function of the waiting time in this case can be calculated as 
\begin{align}
W_q(\theta {\rm{|}}\,case\,3) (\theta)=&X_ \pm ^{(case3)} (\theta)\cdot N_0 \Big( {{\cal L}_{T_1 } (\theta )} \Big) \cdot\nonumber\\
& N_0 \bigg( {{\cal L}_{T_1 } \Big(\lambda  - \lambda {\cal L}_{T_2 } (\theta )\Big)} \bigg)\label{wait3}
,
\end{align}
where, similar to the expression of $\mathbb{E}\{T'\}$,   $\mathbb{E}\{T''\}$ can be illustrated as 
\begin{equation}
\mathbb{E}\{T''\}=-\nabla_\theta {{\cal L}_{T''} (\theta )}\Big|_{\theta=0}. \label{ET}
\end{equation}

 Poisson arrivals see time averages  (PASTA) (i.e., the
fraction of arrivals that see the process in some state is equal to the fraction of time span when the process is in that state) holds for our two-phase M/G/$1$ queueing system \cite{Wolff}, the probabilities that cases $1$, $2$ and $3$ occur can be formulated as  \begin{small}
\begin{align}
{\Pr}\{Case 1\}&=
\frac{\mathbb{E}\{ T_0 \}}{{ \mathbb{E}\{ T_{RC} \} }}=1-\lambda {\mathbb{E}\{T_1\}}-\lambda {\mathbb{E}\{T_2\}} \label{case1Pro},\\
{\Pr}\{Case 2\}&=
\frac{{\mathbb{E}\{ \Gamma \} \mathbb{E}\{ T_1 \} }}{{\mathbb{E}\{ T_{RC} \} }}=\lambda {\mathbb{E}\{T_1\}},
\label{case2Pro}
\\
{\Pr}\{Case 3\}&=
\frac{{\mathbb{E}\{ \Gamma \} \mathbb{E}\{ T_2 \} }}{{\mathbb{E}\{ T_{RC} \} }}=\lambda {\mathbb{E}\{T_2\}},
\label{case3Pro}
\end{align}\end{small}
where $\mathbb{E}\{ T_0 \}$ is defined in \eqref{ET0}.
 

 To keep the stability of the two-stage queueing system,  the busy probability of the relay server must be less than $1$; namely, 
 ${\Pr}\{Case 2\}+{\Pr}\{Case 3\}<1$. Equivalently, it is required that 
\begin{equation}
 \lambda < \frac{1}{\mathbb{E}\{ T_1 \}+\mathbb{E}\{ T_2 \}}. \label{sta}
\end{equation}
must be satisfied.

Finally,   the following proposition can be presented.
\begin{prop}  \label{Theorem3}
In a two-stage queueing system with \textit{i.i.d.} individual  service time span, the mean waiting time of one packet is  
\begin{small}
\begin{align}
\mathbb{E}\{W_q\} =-& {\nabla} _\theta \Bigg( W_q(\theta {\rm{|}}\,case\,1)\cdot{\Pr}\{Case 1\}\Bigg)\Bigg|_{\theta=0}
\nonumber\\ 
&- {\nabla} _\theta \Bigg( W_q(\theta {\rm{|}}\,case\,2)\cdot{\Pr}\{Case 2\}\Bigg)\Bigg|_{\theta=0}\nonumber\\ 
&- {\nabla} _\theta \Bigg( W_q(\theta {\rm{|}}\,case\,3)\cdot{\Pr}\{Case 3\})\Bigg)\Bigg|_{\theta=0}.\label{E1WQ}
 \end{align}
\end{small}
\end{prop}
\begin{proof}
The proof can be found in Appendix \ref{wproof}.

\end{proof}
\textbf{Proposition \ref{Theorem3}} provides a general formula for the mean waiting time, which relates   the   mean waiting time with multiple  queueing system parameters, namely  packet arrival rate,   channel statistics,  and relay service threshold. It is more computationally efficient than  Monte-Carlo simulations and provides guidelines for engineering design.

\section{Power Minimization Problem Formulation}

Based on the prior   analytic packet delay result, in this section we further illustrate the power  minimization problem. Before that, the long-term average power consumption  at different service stages has to be first formulated.  

\subsection{Power Consumption Formulation}
A more general and practical power consumption model is considered.  
Specifically, the overall power consumed by the network is composed of three parts:  energy consumed at the packets, relays, and   AP. The  consumed energy supports: i)  mode transitions between the listening  and receiving modes, ii)  packets receiving in the FSS,  iii) packets  transmitting at the relay   and   receiving at  AP (both in the SSS).  

     { Based on our model, each transceiver is  in   listening,  transmitting, or receiving mode.  Specifically, the transceivers'  states in the idle, FSS, and SSS are listed in Table \ref{States}. }
   \begin{table}[htbp]
\caption{   { Transceivers' States}}
\centering 
\begin{tabular}{c ||  c |   c |   c }
\hline 
\textbf{Stage}                                                                                                             & \textbf{Users' state}  & \textbf{Relay's state}                                                                                                             & \textbf{AP's state}          \\ \hline\hline 
$Idle$                                                                                                                       &  listening &  listening &  listening       \\ \hline
$FSS$ &  listening/transmitting &  receiving &  listening \\ \hline
$SSS$                                                                                                               &  listening  &  transmitting &  receiving \\ \hline 
\end{tabular}\label{States}
\end{table}

In the sequel, we formulate the expected energy consumption during \emph{one regeneration cycle}, based on which the average power consumption is formulated.

\subsubsection{Power Consumption at Users}

The power consumption  in waiting and transmitting modes is   \cite{3}:
\begin{equation}
{P_{u}}=\left\{\begin{array}{ll}
{P_{0,u}}, & \mathrm{waiting\,mode,}\\
{P_{0,u}} + {\vartriangle _{P_u}} {{p_1}}\mathop = \limits^{
\Delta}{P_{T,u}}, & \mathrm{transmitting\,mode,}\label{poweruser}
\end{array} \right.
\end{equation}
where ${P_{0,u}}$ represents the basic electrical circuit  power consumption; ${\vartriangle_{P_u}}$ is the slope of the load-dependent power consumption.  Note that the waiting users are not shut off when they are in the waiting mode, since they  have to listen to the packet-transmitting-request from the relay, as illustrated in Footnote $1$.

As depicted in Fig.~\ref{Regeneration_Cycle},  the  expected overall energy cost at  users in each regeneration cycle contains energy consumed in the idle  mode (with an average duration of $\mathbb{E}\{T_{u}^{idle}\}$) and transmitting mode (with an average duration of $\mathbb{E}\{T_{u}^{t}\}$).  $\mathbb{E}\{T_{u}^{idle}\}$ and $\mathbb{E}\{T_{u}^{t}\}$ can be respectively formulated as 
\begin{align}
\mathbb{E}\{T_{u}^{idle}\}&=\mathbb{E}\{W_q\}-\mathbb{E}\{T_1\},\label{ETUidle}\\
\mathbb{E}\{T_{u}^{t}\}&=   { \mathbb{E}\{T_1\}}. \label{ETUt}
\end{align}

Hence, the   expected  overall energy cost   at the users, denoted by $\mathbb{E}_u^{tot}$, can be given as
\begin{align}
\mathbb{E}_u^{tot} &= \underbrace {\mathbb{E}\{ \Gamma \}\mathbb{E}\{T_{u}^{idle}\}  \cdot P_{0,u} }_{Energy\;cost\;in\;the \; idle\; mode} + \nonumber \\
& \underbrace {\mathbb{E}\{ \Gamma \} \mathbb{E}\{T_{u}^{t}\}  \cdot P_{T,u} }_{Energy\;cost\;for\;data\;transmitting}.\label{Eusertot}
\end{align}

\subsubsection{Power Consumption at the Relay Server}

For relays, the power consumption model    in listening, receiving, and transmitting modes  is  \cite{3} 
\begin{equation}
{P_{R}}=\left\{\begin{array}{ll}
{P_{0,R}} , & \mathrm{receiving\,\,mode,}\\
{P_{0,R}} + {\vartriangle _{P,R}}\cdot{{p_2}}\mathop = \limits^{
\Delta 
}{P_{T,R}}, & \mathrm{transmitting\,\,mode,}\\
{P_{listen,R}}, & \mathrm{listening \,mode,} \label{powerR}
\end{array} \right.
\end{equation}
where ${\vartriangle_{P_R}}$ is the slope of the load-dependent power consumption. Generally, we have ${P_{0,R}}>{P_{sleep,R}}.$

For each regeneration cycle depicted in  Fig.~\ref{Regeneration_Cycle}, the total expected energy cost at the relay server, denoted by $\mathbb{E}_R^{tot}$, includes  energy consumed in the idle period (when relays are in the  listening mode), $k$ FSSs  (when relays are in the receiving mode) and $k$  SSSs  (when relays are in the transmitting mode). Hence, $\mathbb{E}_R^{tot}$ is represented by
\begin{align}
\mathbb{E}_R^{tot}  =& \underbrace {\mathbb{E}\{ T_0 \}  \cdot P_{listen,R} }_{Energy\;cost\;for\;listening}
+
\underbrace {\mathbb{E}\{ \Gamma \} \mathbb{E}\{ T_1 \}  \cdot P_{0,R} }_{Energy\;cost\;for\;data\;receiving} +\nonumber \\
&
\underbrace {\mathbb{E}\{ \Gamma \} \mathbb{E}\{ T_2 \}  \cdot P_{T,R} }_{Energy\;cost\;for\;transmitting}.\label{ERtot}
\end{align}

\subsubsection{Power Consumption at the AP}
 
The power consumption  at the AP in listening and receiving  modes is  described as follows \cite{3}:
\begin{equation}
{P_{AP}}=\left\{\begin{array}{ll}
{P_{0,AP}}, & \mathrm{receiving\,\,mode,}\\
{P_{listen,AP}}, & \mathrm{listening \,mode,} \label{powerR}
\end{array} \right.
\end{equation}
Generally, we have ${P_{0,AP}}>{P_{sleep,AP}}.$ 

For each regeneration cycle,  the listening mode duration at the AP  consists of the idle period and  FSSs for $\Gamma$ packets. The total expected energy cost   at the relay server is 
\begin{align}
\mathbb{E}_{AP}^{tot}  &= \underbrace {\Big(\mathbb{E} \{T_0\}+\mathbb{E}\{ \Gamma \} \mathbb{E}\{ T_1 \}\Big)  \cdot P_{sleep,AP} }_{Energy\;cost\;for\;listening} + \nonumber \\
& \underbrace {\mathbb{E}\{ \Gamma \} \mathbb{E}\{ T_2 \}  \cdot P_{0,AP} }_{Energy\;cost\;for\;data\;receiving}.\label{EAPtot}
\end{align}
 
\subsubsection{Power Consumptions for Mode Transitions}

Mode transitions happen at the relay when i) the   service turns into the FSS  from the idle mode, and ii)  the service  turns into the idle mode from the SSS. Additionally, the  mode transition  happens at the AP when i) the service turns into the SSS  at the completion of the FSS, and ii) when the service turns into the idle/FSS mode at the completion of the SSS.  The switching energy cost is  fixed and denoted by  $E_{sw}$ for each mode transition. For every individual regeneration cycle, the total mode transition energy cost is 
\begin{equation}
\mathbb{E}_{sw}^{tot}=\mathbb{E}\{ 2E_{sw}  + 2kE_{sw} \}\mathop = \limits^{(c')}2E_{sw} (1 + \frac{1}{\pi _0}), \label{vpi_0}
\end{equation}
where  $ 2E_{sw} $ is the energy cost for  mode switching at the relay; $2kE_{sw}$ is the energy cost for  mode switching at the AP; $(c')$ satisfies with \eqref{kpi0}.

It is worth noting that the  electrical circuit power at the relay server  and AP are  much higher than  the packet transmitting power (e.g., $10$ watts versus $0.1$ watts or abound) \cite{3}. Hence, the packet transmitting power   $p_2$  can be ignored.   
 
\subsection{Average Power Minimization Problem}

The power minimization problem under the queueing stability can be formulated as follows:
\begin{align}
&\mathbf{P1:}\;\; \mathop {\min }\limits_N \mathbb{E}\{P_{tot} \} =
\frac{\mathbb{E}_u^{tot}+\mathbb{E}_R^{tot} +\mathbb{E}_{AP}^{tot} 
+\mathbb{E}_{sw}^{tot} }{\mathbb{E}\{T_{RC}\}}
\label{con_objective}\\
&s.t.\quad\quad  \mathbb{E}\{W_q\} \leq D_0\label{delaycon},
\end{align}
where $D_0$ in \eqref{delaycon} is the maximum tolerable mean waiting time.

{ As can be seen from \eqref{n2on}, the resulting expression of  $N_0 (z)$ is in a  recursive form. It is mathematically intractable to obtain the explicit formulae of the  analytical solutions for $N$ and $\pi_0$ due to the following reasons.} 
{Firstly, it is considerably challenging to  determine $n_0 \in \mathbb{N}$ ($n_0\leq n$) at which $N_0(z)$ starts converging.  Secondly, the polynomial terms in  $N_0(z)$ contains $\pi_0$ of different orders.  As implied in  \eqref{n2on}, the highest order of $\pi_0$ when $N_0(z)$  
  converges is $2^{n_0-1}$. Our numerical results in Section \ref{Verify} reveal that $n_0$ can be larger than $6$. The Abel Theorem \cite{ABEL} reveals  that explicit  solution does not exist for irreducible algebraic equations with the highest order being larger than $5$. Thirdly, $\pi_0$ is closely coupled with $N$, as implied in  \eqref{n2on}-\eqref{etaz}. The non-existence of  
the explicit form of $\pi_0$ determines that it is hard to achieve  the closed-form of the global optimal $N$. }

Fortunately,  as illustrated in \textbf{Proposition 1},  $N_0(z)$ converges  for any $|z|<1$. It motivates us to determine the solutions numerically and with a low-level complexity.
The key formulations to obtain  the desired metrics, including $\mathbb{E}\{P_{tot} \}$ in \eqref{con_objective} and $\mathbb{E}\{W_q\}$ in \eqref{delaycon}   are: 1) the expectations of service time span $T_1$ and $T_2$ in \eqref{eTI}; and 2) the expectation of LST in  \eqref{e2q2}.  
It is really complicated  to 
evaluate analytically their expressions.
To obtain the desired metrics, manifold numerical integrations are performed, particular in \eqref{wait2} and \eqref{wait3}.  Straightforward  numerical computations  will lead to an overall  high computational complexity. 

     { The lower bound of the complexity, denoted by $C_{total}$,   is   
\begin{small}
 \begin{align}
C_{total}=&\mathcal{O}\Big(
C_{\mathbb{E}_T}+  \frac{{n(n + 1)}}{2}C_{{\cal L}_T } + {{n(n + 1)}} C_{{\cal L}_T }  + \frac{{n(n + 1)}}{2}C_{{\cal L}_T }  
 \Big)\nonumber\\
  &= {\cal O}\Big(
C_{\mathbb{E}_T}+  2{{n(n + 1)}} C_{{\cal L}_T }\Big),  \label{ctotal} 
 \end{align} 
 \end{small}
 where $ \mathcal{O}(C_{\mathbb{E}_T})$ and  $ \mathcal{O}(C_{\mathcal{L}_{{T}} }) $ are respectively the computational complexity for $\mathbb{E}\{ T_i \}$ in 
\eqref{eTI} and  $\mathcal{L}_{{T}_i}(\theta)$ 
in \eqref{e2q2}. Note that   numerical integrations are executed to compute $\mathbb{E}\{ T_i \}$ and $\mathcal{L}_{{T}_i}(\theta)$, implying that $ \mathcal{O}(C_{\mathbb{E}_T})$ and  $ \mathcal{O}(C_{\mathcal{L}_{{T}} }) $ are  large. 
  The detailed complexity illustrations  
are presented in  Appendix \ref{ccaa}.}

In the next section, we exploit 
Jensen's inequality   to 
approximate these expressions to their alternative closed-forms,  in order  that the computational complexity can be greatly reduced. Then the optimization problem is further illustrated and determined.

\section{Problem Relaxation and Illustration}

\subsection{Performance Bounds}
As can be noticed from \eqref{eTI}, it is mathematically intractable to derive the exact closed-form expression of $\mathbb{E}\{ T_i \}$. Fortunately,   $\mathbb{E}\{T_i\}$   is  strictly convex \emph{w.r.t.}  $x$. Using   Jensen's inequality,  the transmission rate is bounded as \cite{Jensen1}, \cite{Jensen2} 
\begin{align}
J \buildrel \Delta \over = \left\{ \begin{array}{l}
 C_i (h_i ) \le B\log (1 + \frac{{p_i }}{{\sigma _i^2 }}\mathbb{E}\{ X_i \} ) \buildrel \Delta \over = C_i ^{(upper)} (h_i ),  \\ 
 C_i (h_i ) \ge B\log (1 + \frac{{p_i }}{{\sigma _i^2 \mathbb{E}\{ \frac{1}{X_i} \} }})
 \buildrel \Delta \over = C_i ^{(lower)} (h_i )
. \\ 
 \end{array} \right. \label{Jenson_Bound}
\end{align} 

In practice, to guarantee the   QoS of data transmissions, SNR  at the receivers is usually high. Since $\log (1 +x')$ tends to change subtly  in the high $x'$ regime, it follows that $
C_i ^{(lower)} (h_i )\simeq C_i ^{(upper)} (h_i ).$
Thereby, the expectation of $T_i$ can be  approximated by 
\begin{equation}
\mathbb{E}\{ T_i \} \simeq C_i ^{(upper)} (h_i ) \buildrel \Delta \over =\frac{l}{{B\log (1 + \frac{{p_i }}{{\sigma _i^2 }}\sigma _{h_i }^2 )}}
. \label{eTI_appro}
\end{equation}




Due to the irreducible integral computation in   \eqref{e2q2},   it is also  intractable to derive the  exact  closed-form of $\mathcal{L}_{{T}_i}(\theta)$.  
 Let ${\nabla} ^2_x(\cdot)$ represent the second order derivative operator \emph{w.r.t.} $x$. For the sake of clarity, we introduce  $\varsigma (x)$ and define it as 
$\varsigma (x)={\exp \Big( - \theta \frac{l}{C_i (x)}\Big)}.$ 
Then, \eqref{e2q2} can be rewritten as 
 \begin{equation} 
\mathcal{L}_{{T}_i}(\theta)=\int_0^{ + \infty } \varsigma (x) {g_{X_i}(x)}dx. \label{eq33}
\end{equation}

Note that $\theta\geq 0$ in our above formulae.  After some simple calculations, we have  $ {\nabla} ^2_x \big(\varsigma (x)\big)<0$
in the high SNR region,  indicating that $\varsigma (x)$ is concave \emph{w.r.t.} $x$ in the high SNR regime. Additionally, from \eqref{eq33}, we note 
 that $\mathcal{L}_{{T}_i}(\theta)=\mathbb{E}\{ \varsigma (x|\theta)\}$.  Similar to the approximation  method in \eqref{eTI_appro}, $\mathcal{L}_{{T}_i}(\theta)$ can be upper bounded  by exploiting the concavity property, i.e., 
\begin{equation}
\mathcal{L}_{{T}_i}(\theta) \leq \exp \Big( - \theta \frac{l}{{B\log (1 + \frac{{p_i }}{{\sigma _i^2 }}\sigma _{h_i }^2 )}}\Big). \label{LTi_0appro}
\end{equation}

{Following the similar approximation method,  in the high SNR regime,  $\mathcal{L}_{{T}_i}(\theta)$ can be approximated by
\begin{equation}
\mathcal{L}_{{T}_i}(\theta) \simeq \exp \Big( - \theta \frac{l}{{B\log (1 + \frac{{p_i }}{{\sigma _i^2 }}\sigma _{h_i }^2 )}}\Big). \label{LTi_appro}
\end{equation} 
}

With the above mathematical transforms,  $\mathbb{E}\{ T_i \}$  and  $\mathcal{L}_{{T}_i}(\theta)$ are converted into exponential functions,   which are  tractable
closed-form and can be very easily evaluated.  
  More importantly, based on the aforementioned approximations, manifold integral calculations are avoided and the  computational  complexity is significantly reduced. We present the brief complexity analyses in what follows.          {Following the derivations  in  \eqref{ctotal}, we find that  
  the lower bound of the complexity, denoted by $C'_{total}$,  is    
\begin{equation}
C'_{total}= {\cal O}\Big(
c_{\mathbb{E}_T}+  2{{n(n + 1)}} c_{{\cal L}_T }\Big),\label{smallctot} 
\end{equation}
where  
$ \mathcal{O}(c_{\mathbb{E}_T})$ and  $ \mathcal{O}(c_{\mathcal{L}_{{T}} }) $   respectively represent the  complexity for computing     \eqref{eTI_appro} and    \eqref{LTi_appro}.  Compared with $ \mathcal{O}(C_{\mathbb{E}_T})$ and  $ \mathcal{O}(C_{\mathcal{L}_{{T}} }) $, $ \mathcal{O}(c_{\mathbb{E}_T})$ and  $ \mathcal{O}(c_{\mathcal{L}_{{T}} }) $  are considerably small. Hence, the overall complexity is  reduced with our approximation method.}

\subsection{Problem Analysis and Determination} \label{PDsection}

 In this following, we further illustrate that the determination of global optimal $N$ can be obtained numerically by analysing both the objective and  mean delay   constraint. 

 We first present how to derive the feasible range of $N$ that satisfies the mean delay constraint.  Typically, a larger $N$ results in a longer queue length both in the idle and sub-cycle periods. Given that the data transmitting time  on each individual link is fixed, a longer queue length implies that  an end-user experiences more waiting time. Correspondingly, it can be concluded that $\mathbb{E}\{W_q \}$ monotonically increases with $N$, which will also be numerically verified in Subsection \ref{Numer_B}.   


The bisection method  can be applied to obtain the maximum tolerable $N$, hereinafter denoted as $N_{max}$, that adheres to the  delay constraint \eqref{sta}. From \eqref{delaycon},   $N_{max}$ is uniquely determined by
\begin{equation}
\mathbb{E}\{W_q \}\Big|_{ N=N_{max}}\leq D_0,  \;
\rm{and} \;
 \mathbb{E}\{W_q \}\Big|_{ N=N_{max}+1} > D_0. \label{Nnotoptimal}
\end{equation}

 In the following, we further explain that the global optimal $N$ can be obtained efficiently.  Denote the optimum $N$ that minimizes $\mathbb{E}\{P_{tot} \}$ as $N'$. Clearly, the optimum solution of $\mathbf{P1}$ is  
\begin{equation}
N^*=\min\{N_{max}, N'\}.\label{Noptimal}
\end{equation}

To determine $N'$, we first analyse the monotonicity property of $\mathbb{E}\{P_{tot} \}$.  The following three engineering insights and outcomes
  can be achieved:

i)   $\frac{\mathbb{E}_u^{tot}}{{\mathbb{E}\{T_{RC}\}}}$ can be rewritten as 
\begin{align}
 \frac{{\mathbb{E}_u^{tot} }}{\mathbb{E}\{ T_{RC} \}} &=\lambda\mathbb{E} \{ T_u^{idle} \} \cdot P_{0,u} {\rm{ + }}\lambda \mathbb{E}\{ T_u^t \} \cdot P_{T,u} \nonumber \\ 
 & \buildrel \eqref{ETUidle} \over  = \lambda \Big(\mathbb{E}\{ W_q \}  - \mathbb{E}\{ T_1 \} \Big) \cdot  P_{0,u} + \lambda\mathbb{E} \{ T_u^t \} \cdot P_{T,u}.  \label{Eu_tot3}
\end{align}
$\mathbb{E}\{ T_1 \}$ and $\mathbb{E} \{ T_u^t \}$ are constant \emph{w.r.t.} $N$; and $\mathbb{E}\{ W_q \} $ monotonically increases with $N$. In the case that  $P_{0,u}$ takes a sufficiently small value, $ \frac{{\mathbb{E}_u^{tot} }}{\mathbb{E}\{ T_{RC} \}}$  increases very  slightly with $N$; otherwise, it grows fast with $N$. 

ii) 
$\pi_0$ decreases with $N$, resulting in a decreasing  mode switching frequency and further a fall of $\frac{\mathbb{E}_{sw}^{tot}}{{\mathbb{E}\{T_{RC}\}}}$. 
 The illustration will be presented in subsection \ref{Verify} and demonstrated in Fig. \ref{lambda02Theo_Simu} (b).  

iii)  $\frac{\mathbb{E}_R^{tot} +\mathbb{E}_{AP}^{tot}}{\mathbb{E}\{T_{RC}\}}$  can be calculated as
\begin{small}
\begin{align}
\frac{\mathbb{E}_R^{tot} +\mathbb{E}_{AP}^{tot}}{\mathbb{E}\{T_{RC}\}}  \mathop = \limits^{\eqref{E_gamma}, \eqref{E_g2amma}} &  \underbrace {\rho  \cdot P_{sleep,R} 
+ \lambda\mathbb{E}\{ T_1 \}  \cdot P_{0,R}  + \lambda\mathbb{E}\{ T_2 \}  \cdot P_{T,R} \nonumber}_{Power\;consumption\;at\;the\;Relay} \nonumber \\
+ &\underbrace {\Big(  \rho + \lambda\mathbb{E}\{ T_1 \}   \Big)\cdot P_{sleep,AP}+ \lambda\mathbb{E}\{ T_2 \}  \cdot P_{0,AP},}_{Power\;consumption\;at\;the\;AP}  \label{Esim}
\end{align}
\end{small}
where $\rho   =  {{1 - \lambda \mathbb{E}\{ T_1 \}  - \lambda \mathbb{E}\{ T_2 \} }}.$
It can be observed that $\frac{\mathbb{E}_R^{tot} +\mathbb{E}_{AP}^{tot}}{\mathbb{E}\{T_{RC}\}}$ remains  invariant \emph{w.r.t.} $N$.  

Based on the aforementioned analyses, it arises that  there exists  a fundamental trade-off   
between  end-users' power consumption  (i.e., $\frac{\mathbb{E}_u^{tot}}{{\mathbb{E}\{T_{RC}\}}}$) and mode transition power consumption  (i.e., $\frac{\mathbb{E}_{sw}^{tot}}{{\mathbb{E}\{T_{RC}\}}}$).  The monotonicity of $\mathbb{E}\{P_{tot} \}$ is hard to   analyze due to the existence of $\pi_0$.  

It is worthwhile to note that $N_0(z)$ converges for any $|z|<1$ as illustrated in \textbf{Proposition 1}. Then  with the approximation method, 
 tractable closed-formulae can be    derived,    easily evaluated, and hence the overall computational complexity in obtaining $N'$ can be greatly reduced.  
 The resulting table can be stored in the device, which will facilitate  engineering designs.

%

\section{Numerical Results and Discussions}
In this section,   we seek to find the approximated optimum solutions and reveal the mean power consumption and waiting time trade-offs. The performance comparison  between the approximation method and    exact calculations will also be presented.

 The parameters are given in Table \ref{Parameter_Settings}, unless otherwise stated. Note that  packets' transmitting power, i.e., $p_1$ and  $p_2$ take their maximum values to diminish each packet's    transmission time (equivalently, the FSS and SSS  durations are reduced), aiming at saving  the overall electrical circuit energy. 
To obtain insightful conclusions and facilitate the analysis, we normalize   $\sigma _{{h_i}}$ and  assume that $\frac{p_1}{{\sigma^2 _{1}}}=\frac{p_2}{{\sigma^2 _{2}}}$.

\begin{table}[htbp]
\centering
\caption{Parameter Settings}
\begin{tabular}{||c|c||c|c||}
\hline
\textbf{Parameters}                                                                                                             & \textbf{Value} & \textbf{Parameters}                    & \textbf{Value}           \\ \hline\hline 
$p_1=p_2$                                                                                                                       & $0.1$ W            &${P_{0,R}}={P_{0,AP}}$         & $10$                      W \cite{3}\\ \hline
$ {\vartriangle _{P,R}}$ & 2.6 \cite{3}            & ${P_{listen,R}}={P_{listen,AP}}$ & $4$ W                       \cite{3}  \\ \hline
$ p_i/{\sigma^2 _{i}}  (i=1,\; 2)$ & $40$ dB          & $B$                                      & $ 1$  MHz \\ \hline
$ P_{0,u}$                                                                                                                   & $ 0.5$ W            & $E_{sw}$                                  & $8$                       Joules \\ \hline
$ {\vartriangle _{P,u}}$ & $1.2$            & $l$ & $ 10^4$    { bits}     \\ \hline
${{\sigma _{h_1}^2}}$ & $1 $            & ${{\sigma _{h_2}^2}}$ & $ 1 $     \\ \hline
\end{tabular}\label{Parameter_Settings}
\end{table}
\subsection{Verification of the Theoretical Results}\label{Verify}
 $\mathbb{E}\{P_{tot}\}$ and $\pi_0$ curves versus $N$ are respectively depicted in Figs. \ref{lambda02Theo_Simu} (a) and (b). To verify the exact theoretical results and  approximated closed-formulae,   simulation  results  are also presented. 
   {   Note that the exact theoretical results for the average power consumption are  achieved from \eqref{con_objective}, where  $\mathbb{E}\{ T_i \}$   
and  $\mathcal{L}_{{T}_i}(\theta)$ are used and  exactly calculated based on  \eqref{eTI}  and \eqref{eq2}, respectively.   However,  
to efficiently obtain the approximated theoretical results of the average power consumption,  $\mathbb{E}\{ T_i \}$  
and  $\mathcal{L}_{{T}_i}(\theta)$ are  respectively approximated by 
 \eqref{eTI_appro} and  \eqref{LTi_appro}.   
 Additionally,  the idle probability expression is presented in  \eqref{Pi0}. To  exactly calculate its value, $\mathcal{L}_{{T}_i}(\theta)$    as an intermediate parameter, is explicitly calculated based on   \eqref{e2q2}. As a contrast, the approximated  idle probability  is calculated  based on  \eqref{LTi_appro}.}

Additionally,    simulations    are performed by respectively averaging $\mathbb{E}\{P_{tot}\}$ and $\pi_0$ over $10^7$ realizations, including $10^4$ random realizations of  Rayleigh fading channels     {  (as the outer loop in the simulations)} and  $10^3$ consecutive FSSs and SSSs for each channel realization     {  (as the inner loop  in the simulations).}

As can be seen,  theoretical  curves match well with the simulations.  
 Furthermore, it can be noticed that Jensen's bounds in \eqref{eTI_appro} and   \eqref{LTi_appro} 
 are quite tight in the considered high SNR scenario.  
Additionally, we  observe  that, $\pi_0$ decreases with $N$,  
which   is reasonable and agrees with the intuition. The underlying reason is that a larger $N$ indicates a longer data transmission time during the sub-service cycle and hence, more packets join  the waiting queue and wait for their sub-service cycle turn. As a  consequence, there is  a higher probability  that $N_2>0$, indicating a smaller  $\pi_0$ (recall that $\pi_0={\Pr}\{N_2=0\}$). 

\begin{figure}[!htpb]
\centering
\includegraphics[width=0.5\textwidth]{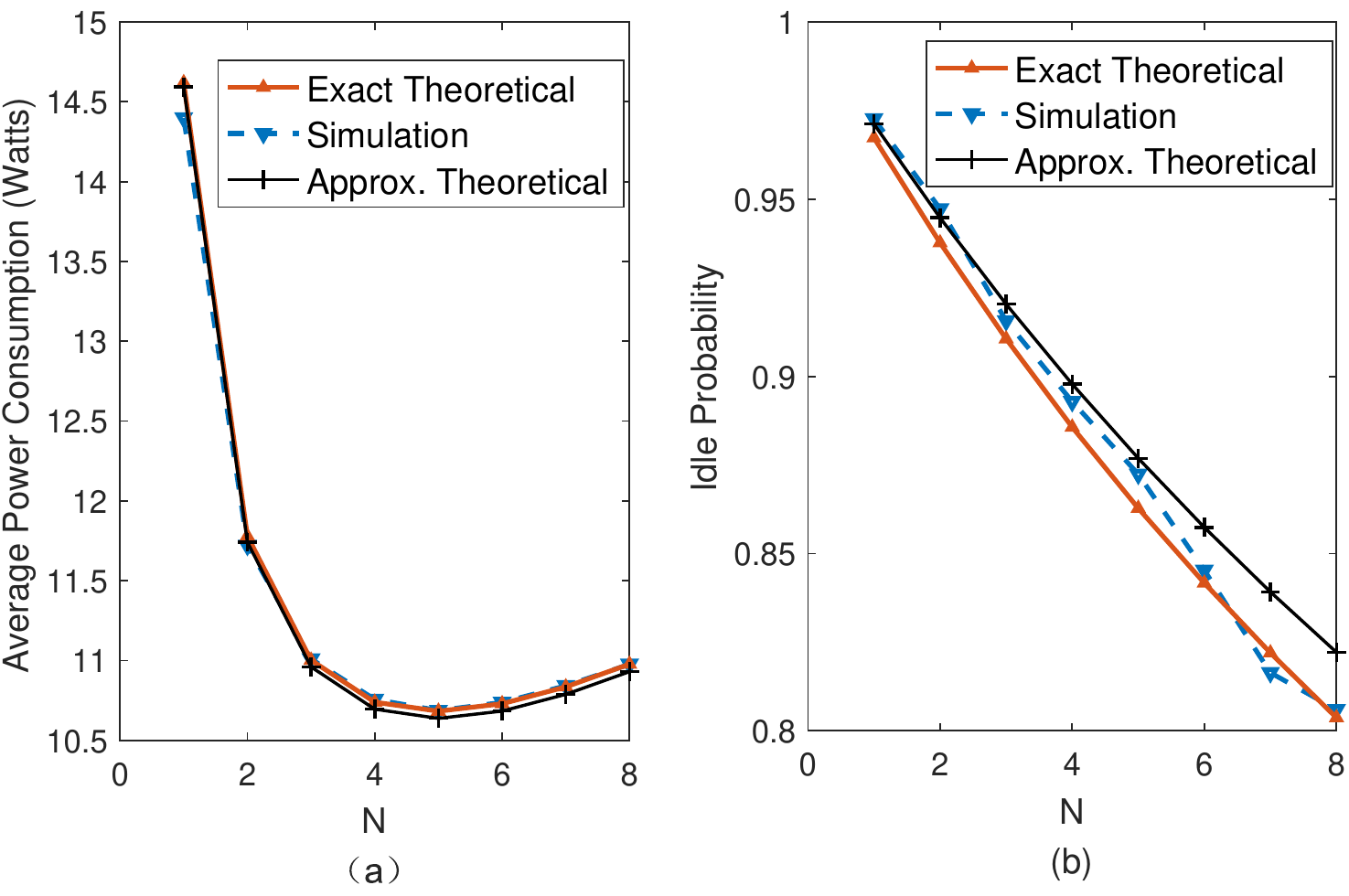}
\caption{{Verification of the theoretical derivations; $\mathbb{E}\{P_{tot}\}$ and $\pi_0$ versus $N$; $\lambda=0.2$.}}\label{lambda02Theo_Simu}
\end{figure}

It is worth noticing that though  $\displaystyle  N_0  ( {\cdot}  )$ presented in \eqref{n2on} appears to include  infinite terms;  however, in the following, we numerically  demonstrate that it converges at a fast rate  under our proposed approach. We take the convergence property of  $N_0  ( {{\cal L}_{T_2 } (\theta )}  )$ in \eqref{Pi0} and \eqref{TPIPP} as an  illustrative example. As indicated in Fig. \ref{Noz_iteration}, $N_0  ( {{\cal L}_{T_2 } (\theta )}  )$ converges when $n=7$ for various $\theta$ values. Note that $N_0  ( {{\cal L}_{T_2 } (0 )}  ) \equiv 1$, which matches with the theoretical results.  This is because $ {{\cal L}_{T_2 } (0 )}\equiv 1$ at $\theta=0$, then     $N_0  ( {{\cal L}_{T_2 } (0 )}  ) \equiv 1$ can be achieved based on the property of the PGF function.
{Additionally, the approximated closed-forms in \eqref{eTI_appro} and  \eqref{LTi_appro} are also considered. The curves  reveal  that   in  the high SNR scenario,  the approximated approach behaves quite similar as the exact formulae.}
 
\begin{figure}[!htpb]
\centering
\includegraphics[width=0.45\textwidth]{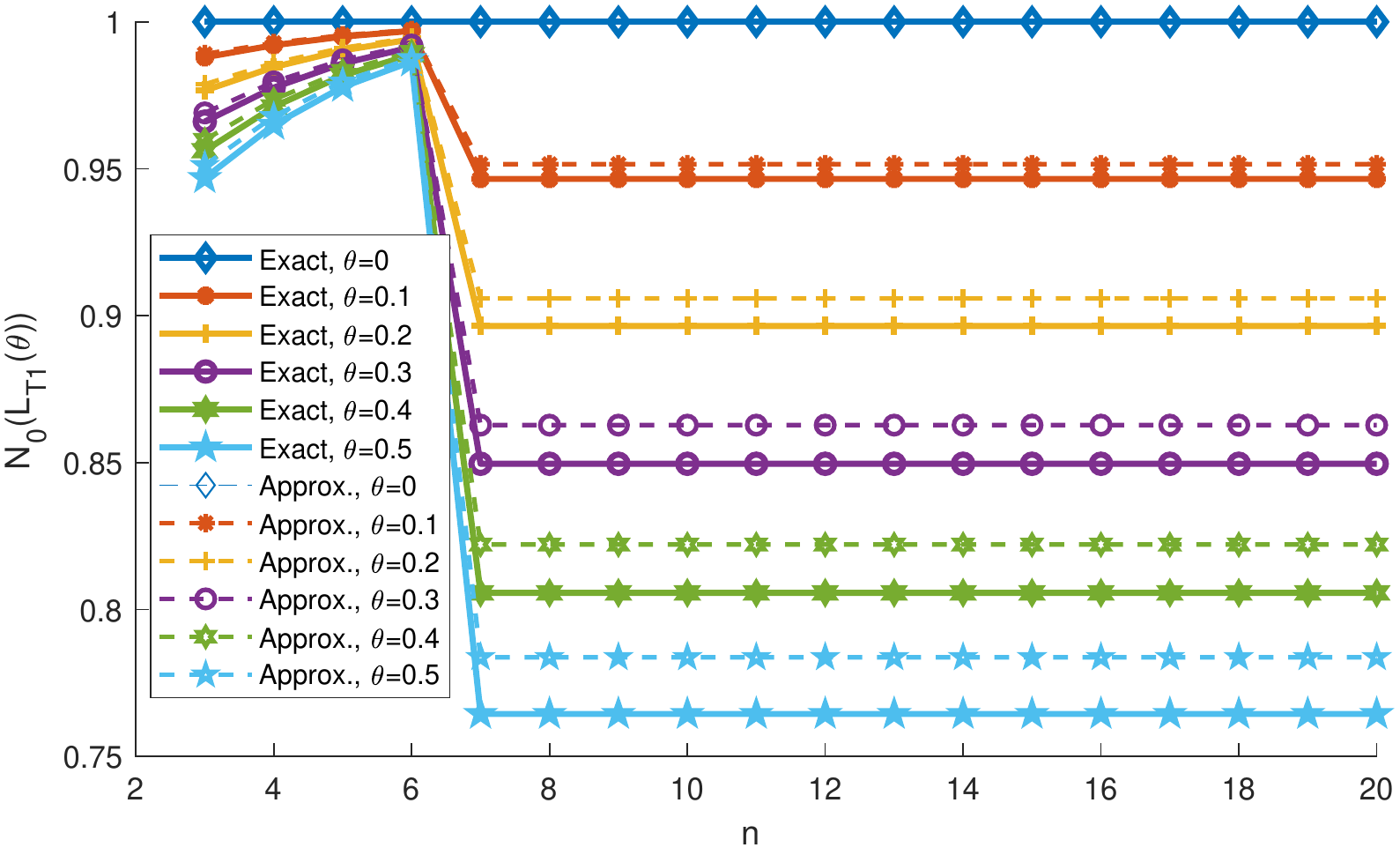}
\caption{{The convergence curves of $N_0 \Big( {{\cal L}_{T_1 } (\theta )} \Big)$ versus $n$; $\lambda=0.2$.}}\label{Noz_iteration}
\end{figure}

\subsection{Impact of $N$} \label{Numer_B}
To investigate the impacts of $N$  on $\mathbb{E}\{P_{tot}\}$ and $\mathbb{E}\{W_q\}$, we plot $\mathbb{E}\{P_{tot}\}$ and $\mathbb{E}\{W_q\}$ curved surfaces  in Figs. \ref{smallPower_3D} and   \ref{smallPower3D_waitingtime}. It can be noticed that  $\mathbb{E}\{P_{tot}\}$ first decreases with $N$ before achieving its  minimum value, and afterwards grows with $N$.  Besides, as demonstrated in Fig. \ref{smallPower3D_waitingtime},  $\mathbb{E}\{W_q\}$ monotonically increases with $N$. Therefore,   sacrificing  the mean waiting time not necessarily brings the benefit of power savings. 
The underlying reason is analyzed as follows.

\begin{figure}[!htpb]
\centering
\includegraphics[width=0.3\textwidth]{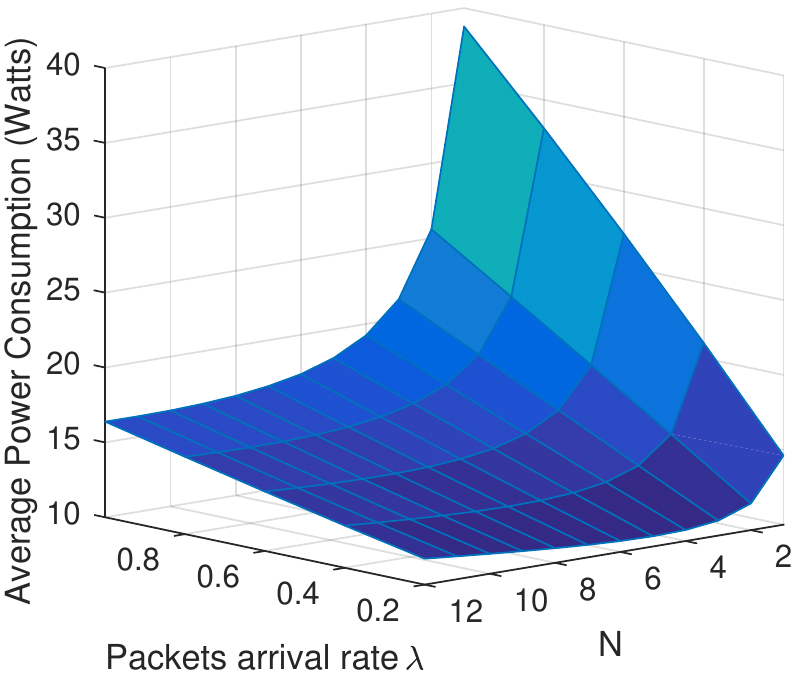}
\caption{$\mathbb{E}\{P_{tot} \}$ versus $\lambda$ and  $N$.}\label{smallPower_3D}
\end{figure} 

Firstly,  in the small  $N$  region, the mean waiting time and $\mathbb{E}\{T_{u}^{idle}\}$ are  small, which   results in a small  $\mathbb{E}_u^{tot}$   value.  
Thereby, $\mathbb{E}\{P_{tot}\}$  is dominated by the power consumed at the AP and relay server. Secondly, as illustrated in subsection  \ref{Verify} and demonstrated in Fig. \ref{lambda02Theo_Simu} (b), $\pi_0$ decreases with $N$, leading to a decreasing  mode switching frequency and   $\frac{\mathbb{E}_{sw}^{tot}}{{\mathbb{E}\{T_{RC}\}}}$. Thirdly, after some simple calculations, it can be revealed that $\frac{\mathbb{E}_R^{tot} +\mathbb{E}_{AP}^{tot}}{\mathbb{E}\{T_{RC}\}}$ in \eqref{con_objective} keeps constant \emph{w.r.t.} $N$. Based on the aforementioned analysis, we claim that the total power consumed at the AP and relay server, i.e., 
$\Big(\frac{\mathbb{E}_{sw}^{tot}}{{\mathbb{E}\{T_{RC}\}}}+\frac{\mathbb{E}_R^{tot} +\mathbb{E}_{AP}^{tot}}{\mathbb{E}\{T_{RC}\}}\Big)$ decreases with $N$. 

However, as $N$ (when $N\geq N^*$) further increases, the mean waiting time grows significantly, as demonstrated in  Fig. \ref{smallPower3D_waitingtime}. Different from the small $N$ case, $\mathbb{E}\{T_{u}^{idle}\}$ takes a larger value as indicated  in  \eqref{ETUidle}. Further, as can be observed    from   \eqref{Eu_tot3}, increasing  $\mathbb{E}\{T_{u}^{idle}\}$ results in a growing $ {\mathbb{E}_u^{tot}} $  such that  $\frac{\mathbb{E}_u^{tot}}{{\mathbb{E}\{T_{RC}\}}}$  dominates the total power consumption instead of $\frac{\mathbb{E}_{sw}^{tot}}{{\mathbb{E}\{T_{RC}\}}}$.
 The above  finally leads to an increasing $\mathbb{E}\{P_{tot} \}$.

\subsection{Impact  of   $\lambda$}

In the sequel, we study the impacts of $\lambda$ on  $\mathbb{E}\{P_{tot}\}$ and   $\mathbb{E}\{W_q\}$.  
In Figs. \ref{smallPower_3D} and   \ref{smallPower3D_waitingtime}, $\mathbb{E}\{P_{tot}\}$ and   $\mathbb{E}\{W_q\}$ versus $\lambda$ are plotted,          { where  $\lambda$ varies  from $0.2$ to $5$}.  Besides, in Fig. \ref{Performance_Comp_Alwayson} (b), the optimal $\mathbb{E}\{P_{tot}\}$  curves  versus $\lambda$ are also presented, where $\frac{p_1}{\sigma_1}=\frac{p_2}{\sigma_2}=40,\, 60$dB, respectively. We have the following two  observations.

\begin{figure}[!htpb]
\centering
\includegraphics[width=0.49\textwidth]{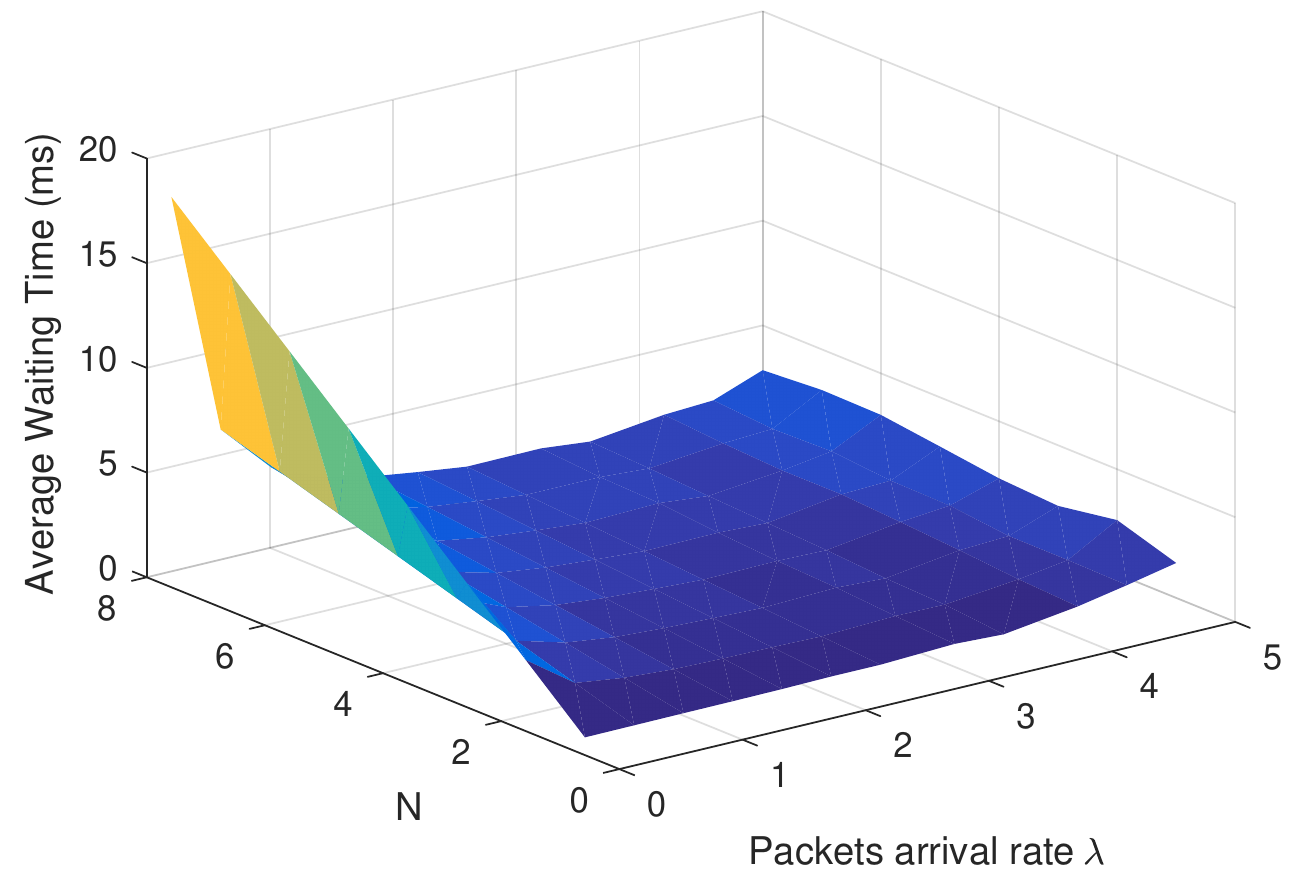}
\caption{$\mathbb{E}\{W_q\}$ surface versus $\lambda$ and  $N$.}\label{smallPower3D_waitingtime}
\end{figure} 

Firstly, the theoretical results show that $\mathbb{E}\{P_{tot}\}$ increases with $\lambda$. These  advocate  the  intuitions and further verify the validity of our derivations.    We first analyse the underlying reason that $\mathbb{E}\{P_{tot}\}$ increases with $\lambda$.  A higher arrival rate means that the relays and   AP are suffering a heavier traffic load, thereby consuming more power. Also, from the users' standpoint, more users are served per unit of time, resulting in higher power consumption.          { We further  illustrate why $\mathbb{E}\{W_q\}$  first decreases   and then increases with   $\lambda$.  As $\lambda$  increases from a small value,  the queue length can easily arrive at the threshold. Hence, the FSS starts in a shorter time, leading to a decreasing $\mathbb{E}\{W_q\}$.   As $\lambda$ further  increases, the queue length becomes larger, and  FSS and SSS last   longer. Correspondingly, it takes a longer time for the user  to be served, even though the FSS has started.}

      {  In Fig. \ref{Trdeoff}, we present  the tradeoffs between the average power consumption  $\mathbb{E}\{P_{tot}\}$ and  $D_0$ when $\lambda=0.2$,
 $0.4$, and $0.6$, respectively.  For the scheme without   optimization (or ``non-Opt" for short), $N=N_{max}$ (see \eqref{Nnotoptimal}). Note that under the   non-Opt scheme,  $\mathbb{E}\{P_{tot}\}$ first sharply decreases  
 with  $D_0$ and then gently increases.  For the optimized results,   $\mathbb{E}\{P_{tot}\}$   does not increase with   $D_0$. It means that  our scheme guarantees  that  sacrificing the delay will bring power saving gains. }
\begin{figure}[!htpb]
\centering
\includegraphics[width=0.35\textwidth]{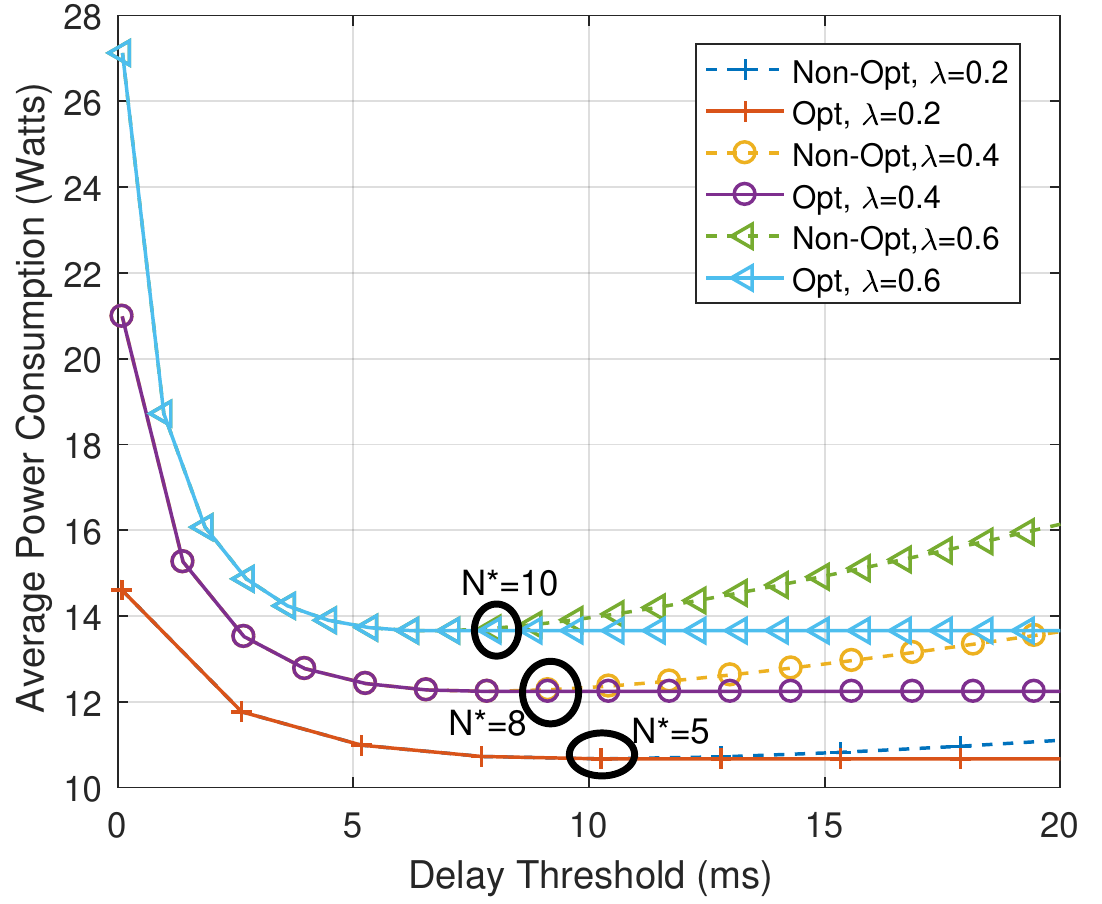}
\caption{     {Tradeoff between the power consumption and   mean waiting time; $D_0=20$ ms.}}\label{Trdeoff}
\end{figure} 

      {Additionally, from Fig. \ref{Trdeoff}, we notice  that in a denser packet arrival   scenario, the energy-saving service threshold  takes a larger value. This is because a larger $\lambda$ implies a smaller packet inter-arrival and shorter time span to be taken to accumulate a certain number of packets.  Hence, for a given  mean waiting time tolerance $D_0$, $N$ is allowed to take a larger value compared with the smaller $\lambda$ scenario. For a certain mean waiting time, according to  \eqref{Eu_tot3}, a larger $N$ does not result in a larger $\frac{\mathbb{E}_u^{tot}}{{\mathbb{E}\{T_{RC}\}}}$ than  that of the smaller $\lambda$ scenario. Meanwhile, based on the analyses in \ref{PDsection} ii) and iii), a larger $N$ leads to a lower $\frac{\mathbb{E}_R^{tot} +\mathbb{E}_{AP}^{tot}}{\mathbb{E}\{T_{RC}\}}$. To rap,   $N^*$ takes a larger value in the denser packet arrival scenario. }
 


\subsection{Performance Comparison with Always-on-Service Policy}

In Figs. \ref{Performance_Comp_Alwayson} (a) and (b), the average  power consumption of the optimum-threshold policy is compared with that of the Always-on-Service (AoS) policy \cite{JL}, where the relay server is activated once one packet arrives. It is clear that the optimal (OP) policy  performs better than the AoS policy in terms of the average  power consumption. 

Additionally, it is demonstrated that for the optimal policy, a larger maximum transmitting power (corresponding to a higher $p_i/\sigma_i$, $\forall i\in \{1,2\}$) results in   reduced average power consumption. This is due to the fact that a larger $p_i/\sigma_i$ ($\forall i\in \{1,2\}$) results in a larger transmitting rate and  further   
 a reduced service time span. Overall power consumption at the relay server and the AP  correspondingly decreases. 
 
Furthermore, numerical results obtained from the approximated closed-form  are also presented. The curves in Fig. \ref{Performance_Comp_Alwayson} (b)    reveal  that     the approximated formulae  get  closer to  the exact  theoretical derivations as  SNR increases. This      coincides with the derivations in \eqref{Jenson_Bound}-\eqref{LTi_appro}.


\begin{figure}[!htpb]
\centering
\includegraphics[width=0.45\textwidth]{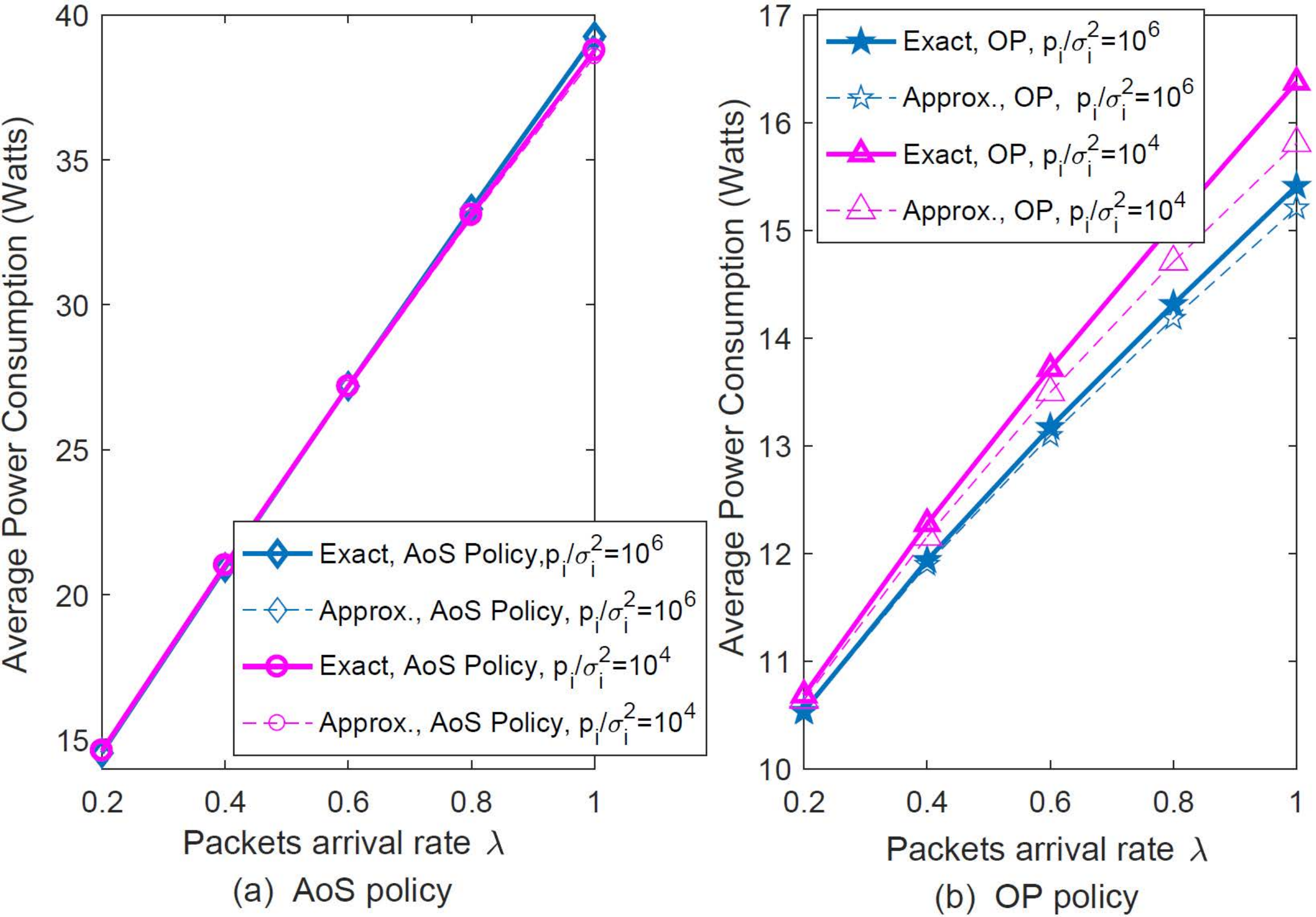}
\caption{{$\mathbb{E}\{P_{tot} \}$ versus $\lambda$ under two different policies; $D_0\leq 10.5$ ms.}}\label{Performance_Comp_Alwayson}
\end{figure}

\section{Conclusions} 

A two-stage, $N$-threshold and gated  M/G/$1$ queueing  communication  has been proposed and investigated for the delay-tolerant  networks.    
 Two important performance metrics, including the mean waiting time and  long-term expected power consumption have been   analysed and associated  with  packet arrival rate, relay service threshold as well as   channel statistics. A more practical expected power consumption was derived in the sense that it included the  electrical circuit energy consumption.    The expected power consumption minimization  problem has been formulated  under the mean delay   constraint. Mathematical approximation methods based on Jensen's inequality were adopted to obtain the tight closed-form bounds, such that computational complexity can be greatly reduced. The above studies and mathematical propositions  will provide guides for practical system design. 
The numerical results  illustrated  
the suitability of Jensen's bounds and   fast convergence rate of  closed-formulae.  The tradeoff between the mean waiting time and power consumptions are revealed. Additionally, 
it has been revealed that: $1$)   in a larger packet arrival rate scenario, the energy-saving   service threshold takes a larger value;  $2$) performance comparison with the always-on-service policy demonstrates the advantage of  our proposed scheme.  


\section {Appendix} 

   { \subsection{Derivations of  \eqref{E_Gamma} and \eqref{E_g2amma}}\label{Wald}}

   { The total time it takes the relay to serve $\Gamma$ packets in the FSSs and SSSs  are respectively $T_{1,1}  + T_{1,2}  +  \cdots  + T_{1,\Gamma } $ and $T_{2,1}  + T_{2,2}  +  \cdots  + T_{2,\Gamma } $ seconds.  Clearly,  $\Gamma$, $T_{1,i}$,  and $T_{2,i}$ 
are all random variables.  $T_{j,1},$ $ T_{j,2}$, $   \cdots$, $ T_{j,\Gamma }$ ($j=1,\, 2$) are identically distributed.  According to Wald's equation, the   following  equation  holds:
 \begin{equation}
\mathbb{E}\{ T_{i,1}  + T_{i,2}  +  \cdots  + T_{i,\Gamma } \}  = \mathbb{E}\{ \Gamma \} \mathbb{E}\{ T_i \},\; i=1,\;2. \label{dejaa}
\end{equation}} 
   { Then based on the delay cycle property in Section $1.2$ of  \cite{Book_queue}, we have  \eqref{E_Gamma}.}
   {  Further, note that 
\begin{align}
T_{RC}  =& T_0  + T_{1,1}  + T_{1,2}  +  \cdots  + T_{1,\Gamma } \nonumber\\
& + T_{2,1}  + T_{2,2}  +  \cdots  + T_{2,\Gamma }. \label{sdw123dejaa}
\end{align}
By combining \eqref{sdw123dejaa} with \eqref{dejaa}, 
we have 
\begin{align}
\mathbb{E}\{ T_{RC}\}   =&\mathbb{E}\{ T_0 \} + \mathbb{E}\{ \Gamma \} \mathbb{E}\{ T_1 \}  + \mathbb{E}\{ \Gamma \} \mathbb{E}\{ T_2 \}\nonumber\\
 \mathop = \limits^{\eqref{ET0}}&\frac{N}{\lambda} + \mathbb{E}\{ \Gamma \} \mathbb{E}\{ T_1 \}  + \mathbb{E}\{ \Gamma \} \mathbb{E}\{ T_2 \} 
 \mathop = \limits^{\eqref{E_Gamma}} \frac{\mathbb{E}\{ \Gamma \}  }{\lambda}.
\label{finalETRC}
\end{align} 
 }

\subsection{Proof of Proposition \ref{Theorem1}}\label{appendixconvexproof}

The following notations are defined:  
\begin{align}
  g_{n}(z)&=\pi _0 \eta \Big( f_n(z) \Big), n \in \mathbb{N}, \label{g_n} \\
   \eta (z) &= 1 - z^N,\label{etaz}\\
    f_0(z)&=z,\\
     f_1(z)&={{\cal L}_{T_2 } (\lambda - \lambda z)},\\
  f_{n} (z)=f_{n-1}\Big(f_1 (z)\Big)&=f_1\Big(f_{n-1} (z)\Big), \; n \in \mathbb{N}^+, \; n\geq 2 \label{iteration},
\end{align}

Under the assumption that $T_1$ and $T_2$  have the same distribution, then 
\begin{equation}
f_{T_1 } (t) = f_{T_2 } (t), \; \rm{and}\; 
{{\cal L}_{T_1 } (z)}= {{\cal L}_{T_2 } (z)} \label{t1t2}.
\end{equation}
According to \eqref{N1Z}-\eqref{N0Z}, we have 
\begin{align}
N_0 (z) \mathop = \limits^{\eqref{N0Z}}& N_0 \Big( {{\cal L}_{T_2 } (\lambda - \lambda z)} \Big)N_0 \Big( {{\cal L}_{T_1 } (\lambda -  \lambda z)\Big)}  - \pi _0 \eta (z)\label{firsteq}\\
\mathop = \limits^{(a)}&  N_0 \Big( {{\cal L}_{T_2 } (\lambda - \lambda z)} \Big)N_0 \Big( {{\cal L}_{T_2 } (\lambda -  \lambda z)\Big)}  - \pi _0 \eta (z) \label{pio}\\
\mathop = \limits^{(b)} &\{N_0 (f_1)\}^2-g_0(z), \label{iter1}\\
\mathop = \limits^{(c)} &\bigg(\bigg [{\Big( {{{\{N_0 (f_{n})\}^2 }}-g_{n  - 1}(z)} \Big)}^2  - g_{n  - 2}(z)\bigg] ^2  -  \cdots   \nonumber \\ 
&- g_{1} (z)\bigg) ^2- g_0(z), \; n \to  \infty\nonumber \\ 
\mathop = \limits^{(d)}&\eqref{n2on}, \label{finalg}
\end{align}   
where $(a)$ is achieved from \eqref{t1t2}; $(b)$ is based on  \eqref{iteration}; $(c)$ is obtained by successively substituting \eqref{iter1} into \eqref{firsteq}; $(d)$  satisfies with the fact that for any $z$ with $|z|<1$ \cite{DB}, 
\begin{equation}
f_n(z) 
 \to z_\infty 
 \equiv  1,  \; \text{as} \; n \to  \infty. \nonumber
\end{equation}
Further, we have 
\begin{equation}
N_0 (f_n) 
 \to   1, \; \rm{and}\; g _n (z)
 \to 0,\; \; \text{as} \; n \to  \infty. \nonumber
\end{equation}

We further illustrate that  $N_0 (z)$ converges absolutely  for any $|z|<1$. When $z=1$,  $f_n (1) \equiv 1\; (\forall n)$ can be obtained from the   recursion formula in  \eqref{iteration}. Further, from \eqref{g_n}, it follows that  $g_n(z) = 0\; (\forall n)$  when $z=1$. Then, it can be  observed  that $N_0 (1)=1$ from $\eqref{n2on}$.  According to the property of PGF \cite{Book_queue}, the convergence of $N_0 (z)$  for any $|z|<1$ is proved.


\subsection{Proof of Proposition  \ref{p00ai}}  \label{appendixconvexproof2} 
By substituting $z=0$ into \eqref{pio}, we have 
\begin{align}
N_0 (0)=&\bigg\{（N_0 \Big( {{\cal L}_{T_2 } (\lambda)} \Big)\bigg\}^2-\pi_0\nonumber\\
\mathop = \limits^{(i)} &   \Pr \{ N_0  = 0\}  \mathop = \limits^{(e)}  0, \label{No01}
\end{align}
where $(i)$ is based on the property of \emph{p.d.f.} function; $(e)$ satisfies with the fact that  $N_0$ is strictly positive, as indicated in \eqref{NEQ_0}. 

Hence, we have  $
 \pi _0=\bigg\{（N_0 \Big( {{\cal L}_{T_2 } (\lambda)} \Big)\bigg\}^2.$
 
\subsection{Proof of Equation \eqref{E1WQ}}  \label{wproof}  
  \begin{proof}
 For easing notations, in the following, we introduce a parameter $ f_k(\theta)$ as follows:
\begin{equation}
 f_k(\theta)= W_q (\theta {\rm{|}}\,case\,k)\cdot\Pr \{ Case\,k\}, \forall k\in\{1,2,3\}. \nonumber
\end{equation}

In the following, the proof is performed according to the property of the LST.  

For  a continuous random variable (referred to as $X$)  with cumulative distribution function $F(t)$, the moments of $X$ can be computed using \cite{qiuqiu} $
{\displaystyle \operatorname {E} [X^{n}]=(-1)^{n}\left.{\frac {{\text{d}}^{n}\{\mathcal{L}_{F}(\theta)\} }{{\text{d}}\theta^{n}}}\right|_{\theta=0},} $
where $\mathcal{L}_{F}(\theta)$  is the Laplace-Stieltjes transforms  of  $F(t)$.

when $n=1$, we have $
{\displaystyle \operatorname {E} [X ]=- \left.{\frac {{\text{d}} \{\mathcal{L}_{F}(\theta)\} }{{\text{d}}\theta }}\right|_{\theta=0}.}$

 Following the property, we have
\begin{equation}
\mathbb{E}\{W_q\} =- {\nabla} _\theta  \Big(  f_1(\theta)\cdot f_2(\theta)  \cdot f_3(\theta) \Big)\Big|_{\theta=0},\nonumber 
\end{equation}  
which can be rewritten as 
\begin{align}  
\mathbb{E}\{ W_q \}  =&  - {\nabla} _\theta  \Big(f_1 (\theta )\cdot f_2 (\theta )\cdot f_3 (\theta )\Big)\Big|_{\theta  = 0} \nonumber \\ 
  = & - {\nabla} _\theta  \Big(f_1 (\theta )\Big)\Big|_{\theta  = 0} {\rm{ }}\cdot  f_2 (0) \cdot f_3 (0)  \nonumber \\ 
& - \nabla _\theta  \Big(f_2 (0)\Big)\Big|_{\theta  = 0} {\rm{ }}\cdot f_1 (0) \cdot f_3 (0)  \nonumber  \\ 
& - {\nabla} _\theta  \Big(f_3 (0)\Big)\Big|_{\theta  = 0}  \cdot f_2 (0){\rm{ }}\cdot f_1 (0 ).\label{adssc}
 \end{align}
Since $ f_k(0)=1$, $\forall k\in\{1,2,3\}$, \eqref{adssc} can be reduced into  \eqref{E1WQ}. 
  \end{proof}
 
          {\subsection{Computational complexity analyses}\label{ccaa}}
         {
As can be observed from  \eqref{waitcase1}-\eqref{wait3}, the computational complexity of  $W_q(\theta {\rm{|}}\,case\,j)$,                                             ($j=1,\, 2,\cdots, 3$) is mainly determined by    $N_0 (\mathcal{L}_{{T}_i}(\theta)) $,  $N_0 \big( {{\cal L}_{T_i } \big(\lambda  - \lambda {\cal L}_{T_j } (\theta )\big)} \big)$,  $(i, j=1, 2)$, ${\cal L}_{T_i} (\theta )$, and $\mathbb{E}\{T'\}$.  Since the value of  ${\cal L}_{T_i} (\theta )$   can be stored and used in the subsequent calculations,  we only take  its    computational complexity  into account once.   Further, note that  $\mathbb{E}\{T'\}$ can be obtained from {$N_0 (\mathcal{L}_{{T}_i}(\theta)) $},    
the   computational complexity  of   
 $\mathbb{E}\{T'\}$ is also ignored. In the following, we first  analyse  computational complexity  for $N_0 (\mathcal{L}_{{T}_i}(\theta)) $.}
 
           {
According to the recursive derivation equations    for  $N_0 (\mathcal{L}_{{T}_i}(\theta)) $, namely \eqref{g_n}-\eqref{finalg},  we find that  the  computational complexity   for $ N_0 \big(\mathcal{L}_{{T}_i}(\theta)\big)$, denoted by   
$C_c\Big(N_0 \big(\mathcal{L}_{{T}_i}(\theta)\big)\Big) $,   is determined by  
  $C_c\Big(f_i\big(\mathcal{L}_{{T}_i}(\theta) \big)\Big)$,  $(i=0,1,..., n-1)$.  Note  that $
  C_c\Big(f_1\big(\mathcal{L}_{{T}_i}(\theta) \big)\Big)=\mathcal{O}(C_{\mathcal{L}_{{T}} })$. 
Further,  based on the recursive derivations,  we have                                                                      the general formula as $
  C_c\Big(f_n\big(\mathcal{L}_{{T}_i}(\theta) \big)\Big)=\mathcal{O}(nC_{\mathcal{L}_{{T}} }).$}
         {As a result,  
 \begin{small}
 \begin{align}
C_c\Big(N_0 \big(\mathcal{L}_{{T}_i}(\theta)\big)\Big) &=\mathcal{O}\Big((1 + 2 +  \cdots  + n) \cdot  C_{{\cal L}_T }\Big)\nonumber\\
&= {\cal O}\Big(\frac{{n(n + 1)}}{2} C_{{\cal L}_T } \Big).  \label{CCC} 
 \end{align} 
 \end{small}}                                                        
           {Further, following the similar analyses, 
  the      computational complexity  of  $N_0 \big( {{\cal L}_{T_i } \big(\lambda  - \lambda {\cal L}_{T_j } (\theta )\big)} \big)$ can be written  as follows:
 \begin{small} \begin{align}
   C_c\Big ( N_0 \big( {{\cal L}_{T_i } \big(\lambda  - \lambda {\cal L}_{T_j } (\theta )\big)}\big ) \Big ) & = {\cal O}\Big((1 + 2 +  \cdots  + n) \cdot 2C_{{\cal L}_T } \Big)\nonumber\\
&=  {\cal O}\Big({{n(n + 1)}} C_{{\cal L}_T } \Big).          
   \end{align}\end{small} }
         { 
 Following  the derivations of \eqref{CCC},  the computational complexity for $\pi _0$                                                                  can be  achieved as ${\cal O}\Big(\frac{{n(n + 1)}}{2} C_{{\cal L}_T } \Big)$.  Finally, \eqref{ctotal} can be obtained.   }  
 
\balance

\end{document}